\documentclass{article}
\usepackage[hmargin=1in,vmargin=1.25in]{geometry}
 \usepackage{authblk}
\usepackage[utf8]{inputenc}
\usepackage{amsmath,amsthm, amssymb, amsfonts, amsbsy,mathtools}
\usepackage[mathscr]{euscript}
\usepackage{xspace}
\usepackage{graphicx}
\usepackage{color}
\usepackage{enumerate}
\usepackage[english]{babel}
\usepackage[numbers,sort&compress]{natbib}
\usepackage{verbatim}
\usepackage[usenames,dvipsnames]{xcolor}
\usepackage{url}
\usepackage{booktabs}
\usepackage{caption}
\usepackage{subcaption}

\usepackage{pgf, tikz}
\usetikzlibrary{arrows, automata}
\usetikzlibrary{fit}
\usetikzlibrary{shapes.geometric}
\usepackage{ifthen}
\usetikzlibrary{calc}

\usetikzlibrary{patterns}

\usepackage{subcaption}
\usepackage{pgfplotstable}
\usepackage{pgfplots}
\usetikzlibrary{pgfplots.groupplots}
\usepgfplotslibrary{colorbrewer}
\pgfplotsset{compat = newest} 
\pgfplotsset{cycle list/Set1-8} 
\usetikzlibrary{pgfplots.statistics, pgfplots.colorbrewer} 
\usepackage{perf}

\usepackage{algorithm}
\usepackage[noend]{algpseudocode}

\usepackage{threeparttable}

\usepackage[textsize=tiny]{todonotes}
\newcommand{\pmoura}[1]{\todo[backgroundcolor=orange!40]{P: #1}}

\usepackage{cleveref}
\usepackage{soul}
\newcommand{\hyaman}[1]{\todo[backgroundcolor=olive!40]{H: #1}}
\newcommand{\mdavari}[1]{\todo[backgroundcolor=green!40]{M: #1}}
\algrenewcommand\algorithmiccomment[1]{%
  {\footnotesize\hfill$\triangleright$ #1}
}
\algrenewcommand\algorithmicrequire{\textbf{Input:}}
\algrenewcommand\algorithmicensure{\textbf{Output:}}

\theoremstyle{plain}
\newtheorem{theorem}{Theorem}
\newtheorem{proposition}[theorem]{Proposition}

\newtheorem{lemma}[theorem]{Lemma}

\newtheorem{conjecture}{Conjecture}
\newtheorem{problem}{Problem}

\newcommand{\R}{\mathbb{R}}

\newcommand{\Z}{\mathbb{Z}}

\newcommand{\Q}{\mathbb{Q}}
\newcommand{\true}{\textsc{true}}
\newcommand{\false}{\textsc{false}}
\newcommand{\opt}{\textsc{opt}}
\newcommand{\NP}{\mathsf{NP}}

\newcommand{\bigO}{\mathcal{O}}
\newcommand{\eps}{\varepsilon}

\DeclareMathOperator{\poly}{poly}
\newcommand{\UB}{\mathrm{UB}}

\newcommand{\pcst}{\textsc{pcst}}
\newcommand{\bpcst}{\textsc{bpcst}}
\newcommand{\mwcs}{\textsc{mwcs}}

\newcommand{\phablo}[1]{{\color{magenta} [Phablo: #1]}}

\makeatletter
\let\Rset\R
\DeclareRobustCommand{\R}{\@ifnextchar2{\Rrevision}{\Rset}}
\def\Rrevision2#1{{#1}}
\makeatother

\renewcommand{\proof}[1]{\par\addvspace{\topsep}\noindent\textit{#1}\enspace\ignorespaces}
\renewcommand{\endproof}{\par\addvspace{\topsep}}
\newcommand{\Halmos}{\mbox{\quad$\square$}}

\AddToHook{file/perf-profile.tex/before}{\begingroup\pgfplotsset{width=\linewidth}}
\AddToHook{file/perf-profile.tex/after}{\endgroup}
\AddToHook{file/perf-profile-large-k.tex/before}{\begingroup\pgfplotsset{width=\linewidth}}
\AddToHook{file/perf-profile-large-k.tex/after}{\endgroup}
\AddToHook{file/perf-profile-large-n.tex/before}{\begingroup\pgfplotsset{width=\linewidth}}
\AddToHook{file/perf-profile-large-n.tex/after}{\endgroup}
\title{{Balanced spanning forests} of  graphs: Hardness and solving methods}

\author[1]{Morteza Davari}
\author[2]{Phablo F. S. Moura$^{*}$}
\author[3]{Hande Yaman}
\affil[1]{\R2{SKEMA Business School}, Universit\'{e} C\^{o}te d'Azur, Lille, France\\
\texttt{\normalsize morteza.davari@skema.edu}}
\affil[2]{Research Centre for Operations Research and Business Statistics (ORSTAT), KU Leuven, Belgium\\
Institute of Computing, University of Campinas, Brazil\\
\texttt{\normalsize phablo.moura@kuleuven.be; pfsm@unicamp.br}\\
$^{*}$Corresponding author}
\affil[3]{Research Centre for Operations Research and Business Statistics (ORSTAT), KU Leuven, Belgium\\
\texttt{\normalsize hande.yaman@kuleuven.be}}
\date{October 23, 2026}

\begin{document}
\maketitle
\begin{abstract}

The balanced connected $k$-partition  problem (\textsc{bcp}) is a classic  problem, which consists in partitioning the set of vertices of a vertex-weighted connected  graph into a collection of~$k$ classes such that each class induces a connected subgraph of \emph{roughly} the same weight.
In this study, we investigate edge-weighted variants of $\textsc{bcp}$, where, 
given a connected graph $G$, \R2{ a positive integer} $k 
$, \R2{and edge weights,} 
the goal is to compute a spanning $k$-forest~$\mathcal{T}$ of $G$ (i.e., a forest with exactly $k$ trees) that minimizes the weight of the heaviest tree in~$\mathcal{T}$ in the min-max version, or maximizes the weight of the lightest tree in~$\mathcal{T}$ in the max-min version.
We show that both versions of this problem are $\NP$-hard on complete graphs with $k=2$, unweighted split graphs, and unweighted bipartite graphs with $k\geq 2$ fixed.
Moreover, we prove that these problems do not admit subexponential-time algorithms, unless the Exponential-Time Hypothesis fails.
We focus on the min-max version and devise a tight $k$-approximation algorithm, compact and \R2{non-compact} integer {linear programming} formulations, branch-and-cut, and branch-and-price algorithms. Finally, we present the outcomes of an experimental study on the performance of different solution methods.
The source code of the complete implementation of the proposed algorithms is also available.

\end{abstract}



\section{Introduction}\label{sec:intro}


Let $[k]$ denote the set $\{1, 2, \ldots, k\}$, for every
integer~$k \geq 1$.
{Let $\R2{\Z_{\geq 0}}$ and $\R2{\Q_{\geq 0}}$ denote the sets of \R2{non-negative} integer and rational numbers, respectively.}
{Let $G$ denote an undirected graph, and let $V(G)$ and $E(G)$ denote the set of vertices and edges of $G$, respectively.}
{For simplicity, we also use the shorthand notation $V$ and $E$ when the graph $G$ is clear from the context.}
A \textit{connected $k$-partition} of~{$G$} is a partition of~{$V(G)$} into \R2{non-empty}
classes~$\{V_i\}_{i=1}^k$ such that, for each~$i \in [k]$, the
subgraph~$G[V_i]$ is connected, where~$G[V_i]$ denotes the subgraph
of~$G$ induced by the set of vertices~$V_i$.
Given a connected graph $G$, a positive integer $k$, the \textsc{Min-Max Balanced Connected Partition} problem (\textsc{min-max-bcp}) consists in computing a connected $k$-partition of $G$ that minimizes the {cardinality} of {a} largest class in this partition, that is, it minimizes $\max_{i \in [k]} |V_i|$.
In a more general setting, one is also given a weight function $w \colon V(G) \to \R2{\Q_{\geq 0}}$ and aims to find a connected \R2{$k$-partition} that minimizes the weight of {a} heaviest class, that is, $\max_{i \in [k]} \sum_{v \in V_i} w(v)$. 
The \textsc{Max-Min Balanced Connected Partition} problem (\textsc{max-min-bcp}) is defined analogously: Given a connected graph~$G$, a positive integer $k$, and $w\colon V(G)\to \R2{\Q_{\geq 0}}$, the goal is to find a connected $k$-partition $\{V_i\}_{i=1}^k$ that maximizes $\min_{i \in [k]} \sum_{v \in V_i} w(v)$.

There is a vast literature on both versions of this problem that includes computational complexity~\citep{Chl96,ChaSalWak07}, approximation algorithms~\citep{CasEtAl21,ChenCLXZ21,moura_balanced_2023}, exact solving methods based on integer programming~\citep{MIYAZAWA2021826,ZhoWanDinHuSha19}, and a polyhedral study~\citep{miyazawa2020cut}. 
Furthermore,  a rich collection of real-world applications in image processing~\citep{LucPerSim93}, cluster analysis~\citep{MarSimNal97}, robotics~\citep{ZhoWanDinHuSha19},  power networks~\citep{10288739, wang2024layout}, public security~\citep{MIYAZAWA2021826}, and political districting~\citep{mendes2022capacitated} can be viewed as instances of (\textsc{max-min} or \textsc{min-max})-\textsc{bcp}.

In this paper, we investigate edge-weighted variants of $\textsc{bcp}$, where we are given a connected graph $G$, a positive integer $k$, and an edge-weight function $w \colon E(G)\to\R2{\Q_{\geq 0}}$, and the goal is to compute a spanning $k$-forest~$\mathcal{T}$ of $G$ (i.e., a forest with exactly $k$ trees) that minimizes the weight of {a} heaviest tree in~$\mathcal{T}$ in the min-max version, and maximizes the weight of {a} lightest tree in~$\mathcal{T}$ in the max-min version.
More formally, these problems are defined as follows.

\begin{problem}[\textsc{Min-Max Balanced Spanning Forest} (\textsc{min-max-bsf})]\hfill\\
\textsc{Instance:} Connected {undirected $n$-vertex} graph $G$, {positive integer} $k$, and  $w \colon E(G) \to \R2{\Q_{\geq 0}}$. \\
\textsc{Objective:} Collection of {exactly} $k$ vertex-disjoint trees $\{T_i\}_{i \in [k]}$ spanning $V(G)$ that minimizes 
\(\max_{i \in [k]} w(T_i)\), where $w(T_i):= \sum_{e \in E(T_i)} w(e)$ for all $i \in [k]$.
\end{problem}

The \textsc{Max-Min Balanced Spanning Forest} (\textsc{max-min-bsf}) is defined analogously. 
If $k=1$, these problems are precisely the classic \textsc{Minimum Spanning Tree} and \textsc{Maximum Spanning Tree}, where the latter can be transformed into the former simply by reversing the signs of the weights~\citep[see][p.~336]{pemmaraju2003computational}.
For $k\ge 2$, \textsc{max-min-bsf} and \textsc{min-max-bsf} are not equivalent.
Precisely, an optimal solution to one of the problems is not necessarily an optimal solution to the other, as illustrated in Figure~\ref{fig:example:optimal}.
This contrasts with the fact that \textsc{min-max-bcp} and \textsc{max-min-bcp} are equivalent if $k=2$, that is, for any instance, an optimal solution to one problem is also an optimal solution to the other (possibly with different optimal values).

\begin{figure}[t!]
\centering
\begin{subfigure}{.32\linewidth}
    \centering
    \begin{tikzpicture}[
            auto,
            node distance = 1.3cm, 
            semithick, 
            scale=0.6, 
            every node/.style={scale=0.75}
        ]

        \tikzstyle{every state}=[
            draw = black,
            semithick,
            fill = white,
            minimum size = 2mm,
            inner sep=0pt
        ]


        \node[state] (v5) []{};
        \node[state] (v2) [above left of=v5] {};
        \node[state] (v3) [below left of=v5] {};
        \node[state] (v1) [left of=v2] {};
        \node[state] (v4) [left of=v3] {};
        \node[state] (v6) [right of=v5] {};
        \node[state] (v7) [above right of=v6] {};
        \node[state] (v8) [below right of=v6] {};

        \path (v1) edge node{1} (v2);
        \path (v2) edge node{1} (v5);
        \path (v5) edge node{1} (v3);
        \path (v3) edge node{1} (v4);
        \path (v2) edge node{3} (v3);
        \path (v5) edge node{1} (v6);
        \path (v6) edge node{1} (v7);
        \path (v8) edge node{1} (v6);
    \end{tikzpicture}%
    \caption{Edge-weighted graph.\\\hfill} 
\end{subfigure}\hfill
\begin{subfigure}{.32\linewidth}
    \centering
    \begin{tikzpicture}[
            auto,
            node distance = 1.3cm, 
            semithick, 
            every node/.style={scale=0.75}
        ]

        \tikzstyle{every state}=[
            draw = black,
            semithick,
            fill = white,
            minimum size = 2mm,
            inner sep=0pt
        ]

         \node[state] (v5) []{};
        \node[state] (v2) [above left of=v5] {};
        \node[state] (v3) [below left of=v5] {};
        \node[state] (v1) [left of=v2] {};
        \node[state] (v4) [left of=v3] {};
        \node[state] (v6) [right of=v5] {};
        \node[state] (v7) [above right of=v6] {};
        \node[state] (v8) [below right of=v6] {};

        \path (v1) edge node{1} (v2);
        \path (v2) edge node{1} (v5);
        \path (v5) edge node{1} (v3);
        \path (v3) edge node{1} (v4);
        \path (v6) edge node{1} (v7);
        \path (v8) edge node{1} (v6);
    \end{tikzpicture}%
    \caption{Optimal solution for \textsc{min-max-bsf} of value~4.}
\end{subfigure}\hfill
\begin{subfigure}{.32\linewidth}
    \centering
    \begin{tikzpicture}[
            auto,
            node distance = 1.3cm, 
            semithick, 
            scale=0.6, 
            every node/.style={scale=0.75}
        ]

        \tikzstyle{every state}=[
            draw = black,
            semithick,
            fill = white,
            minimum size = 2mm,
            inner sep=0pt
        ]

        \node[state] (v5) []{};
        \node[state] (v2) [above left of=v5] {};
        \node[state] (v3) [below left of=v5] {};
        \node[state] (v1) [left of=v2] {};
        \node[state] (v4) [left of=v3] {};
        \node[state] (v6) [right of=v5] {};
        \node[state] (v7) [above right of=v6] {};
        \node[state] (v8) [below right of=v6] {};

        \path (v1) edge node{1} (v2);
        \path (v3) edge node{1} (v4);
        \path (v2) edge node{3} (v3);
        \path (v5) edge node{1} (v6);
        \path (v6) edge node{1} (v7);
        \path (v8) edge node{1} (v6);
    \end{tikzpicture}%
    \caption{Optimal solution for \textsc{max-min-bsf} of value~3.}
\end{subfigure}
    \caption{Example of edge-weighted graph which shows that an optimal solution to \textsc{min-max-bsf} is not necessarily an optimal solution to  \textsc{max-min-bsf} with $k=2$.}
    \label{fig:example:optimal}
\end{figure}


The study of \textsc{min-max-bsf} is originally motivated by a problem emerging in a computational topology software, which performs expensive computations organized in graph structures {\citep[see][]{madkour2017}}.
\citet{bakhshipour2019} propose heuristics to another application of \textsc{min-max-bsf} in the context of the design of decentralized urban drainage systems, where a feasible sewer network layout is represented as a collection of rooted trees, 
each root corresponding to an outlet connected to, for example, a wastewater treatment plant.

\citet{madkour2017} proved that the decision version of \textsc{min-max-bsf} with $k=2$ is weakly $\NP$-complete on bipartite graphs.
\citet{vaishali2018} claimed that the \emph{unweighted} version of \textsc{min-max-bsf} is $\NP$-complete on bipartite graphs when~$k=2$ using a reduction from \textsc{2-partition}. 
However, their reduction is not polynomial, as the size of the constructed graph is larger than the sum of the numbers in the \textsc{2-partition} instance.
On the positive side, \citet{vaishali2018} designed an algorithm that solves the problem restricted to trees in $\bigO(kn^3)$ time.

\subsection*{Contributions}
This work offers theoretical and practical contributions to the study of balanced spanning forest problems.
In Section~\ref{sec:hardness}, we show that both \textsc{max-min-bsf} and~\textsc{min-max-bsf} are $\NP$-hard on complete graphs with $k=2$, unweighted split graphs, and unweighted bipartite graphs with $k\geq 2$ fixed.
Moreover, we prove that these problems do not admit subexponential-time algorithms, unless the Exponential-Time Hypothesis fails. 
We design a tight $k$-approximation algorithm and heuristics for \textsc{min-max-bsf} in Section~\ref{sec:heur}. 

On the practical side, we first derive three mixed-integer {linear programming} formulations for each version of~\textsc{bsf} in Section~\ref{sec:formulations}: a compact model based on network flows, an assignment formulation that uses \R2{subtour-elimination} constraints, and a natural \R2{set-partitioning} formulation.
{The compact formulation is complemented by valid inequalities that strengthen its linear-programming relaxation. For the assignment formulation, we develop a dedicated branch-and-cut algorithm.}

In Section~\ref{sec:bp}, we describe a branch-and-price approach based on {the set-partitioning formulation. The balancing objective gives rise to a pricing problem with a \R2{non-linear} structure that cannot be handled efficiently through a straightforward linearization. We therefore develop tailored solution strategies via repetitively solving \R2{a} budgeted variant of the Prize-Collecting Steiner Tree problem. These strategies enable an effective exact solution method for the balanced spanning forest setting.}

{More broadly, the modeling and algorithmic ideas developed for these balancing objectives may be useful in other network optimization problems with min-max or max-min objective functions. In particular, balancing objectives may give rise to similar \R2{non-linear} pricing problems that cannot be handled effectively through straightforward linearization, and the valid inequalities developed for the compact formulation may help strengthen related models. Thus, the proposed techniques extend beyond balanced spanning forests to other problems for which standard min-sum formulations and decomposition methods do not apply directly.}

In Section~\ref{sec:experiments}, we report on computational experiments for the \textsc{min-max-bsf} with three solving methods based on the proposed formulations and heuristics.
These results indicate that the branch-and-price approach is superior to the others, particularly on the instances with $k>2$.
Concluding remarks and directions for further research are discussed in Section~\ref{sec:conclusion}.
An open-source repository with the complete implementation of all proposed solving methods is available\footnote{ \texttt{https://github.com/md6712/Edge-Weighted-Balanced-Connected-Partitions}}.


\section{Hardness
}\label{sec:hardness}

Let us first present the definition of the 2-partition problem, which is used in the next reductions. {We use~$\R2{\Z_{>0}}$ to denote the set of positive integers.} 
\begin{problem}[\textsc{2-partition}]\hfill\\
\textsc{Instance:} Set $A$ and a function $s \colon A \to \R2{\Z_{>0}}$ such that $s(a) < S$ for each $a \in A$ and $\sum_{a \in A} s(a) = 2S$. \\
\textsc{Question:} Is there a subset $A'$ of $A$ such that $\sum_{a \in A'} s(a) = S$?
\end{problem}

\begin{theorem}
    The decision versions of \textsc{min-max-bsf} and \textsc{max-min-bsf} are weakly $\NP$-complete on complete graphs with $k=2$.
\end{theorem}
\proof{Proof.}
{The decision versions of \textsc{min-max-bsf} and \textsc{max-min-bsf} belong to $\NP$: given a subgraph of the input graph as a certificate, one can verify in polynomial time that it is indeed a spanning $k$-forest and that its weight satisfies the required bound.}

Consider an instance~$(A, s)$ of \textsc{2-partition}.
We now construct a complete graph~$G$ with vertex set $\{v_a : a \in A\} \cup \{x,y\}$,  and define edge weights $w$ such that $w(e)=s(a)$ if $e=\{x,v_a\}$ or $e=\{y,v_a\}$ with $a \in A$, and $w(e)=S+1$ otherwise.

Let $\{T_1,T_2\}$ be a spanning $2$-forest of $G$ such that $w(T_1)\le w(T_2)\le S$.
Clearly, both trees can contain only edges with one endpoint in $\{x,y\}$ and the other in $V(G)\setminus\{x,y\}$.
Moreover, $x$ and $y$ cannot belong to the same tree, say~$T_i$, otherwise $T_{3-i}$ is a singleton and $w(T_i) > S$.
Therefore, it holds that $w(T_1)=w(T_2)=S$, and so the leaves of any of these trees induce a set~$A'\subset A$ such that $\sum_{a \in A'} s(a) = S$.

Similarly to the min-max case, we can reduce an instance $(A,s)$ of \textsc{2-partition} to a complete graph $G$ on vertices $\{v_a : a \in A\} \cup \{x,y\}$  with edge weights $w$ such that $w(e)=s(a)$ if $e=\{x,v_a\}$ or $e=\{y,v_a\}$ with $a \in A$, and $w(e)=\varepsilon:=\min_{a \in A} s(a)/|E(G)|$ otherwise.

Let $E':=\{e\in E(G) :  w(e) > \varepsilon\}$, and let $\{T_1,T_2\}$ be a spanning $2$-forest of $G$ such that $w(T_1)\ge w(T_2)\ge S$.
First note that $x$ and $y$ cannot belong to the same tree $T_i$ with $i \in \{1,2\}$, otherwise $w(T_{3-i}) < \varepsilon|E(G)| < S$.
Suppose to the contrary that one of the trees contains an edge in $E(G)\setminus E'$.
Since $T_1$ and $T_2$ are trees that induce a partition of~$V(G)$, there exists $b \in A$ such that neither $\{v_b,x\}$ nor $\{v_b,y\}$ belong to any of these trees.
Hence, $w(T_1)+w(T_2) < w(E(T_1)\cap E') + w(E(T_2)\cap E') + s(b) \le 2S$, a contradiction. 
Therefore, all the edges of $T_1$ and $T_2$ are in $E'$, $w(T_1)=w(T_2)=S$, and so the leaves of any of these trees induce a set~$A'\subset A$ such that $\sum_{a \in A'} s(a) = S$.

{The result follows since \textsc{2-partition} is weakly $\NP$-complete \citep[see][]{garey1979computers}.}
\Halmos \endproof



The previous result implies that \textsc{min-max-bsf} and \textsc{max-min-bsf} are \emph{weakly} $\NP$-hard on complete graphs, that is, they are computationally difficult if the edge weights are not bounded by a polynomial in the size of the graph. 
 We next investigate the complexity of the balanced spanning forest problems on \emph{unweighted} graphs, that is, graphs with all edges having weight equal to one.
 We first show that both versions of the problem are $\NP$-hard 
on split graphs, a class that includes the complete graphs.
A \emph{split graph} is a graph in which the vertices can be partitioned into a clique and an independent set.
Split graphs are \emph{chordal graphs},  a class of graphs that is characterized by a linear ordering of its vertices such that, for each vertex $v$,  the neighbors of $v$ that occur after $v$ in this ordering induce a clique.
The proposed reductions are from the 3-partition problem defined as follows.

\begin{problem}[\textsc{3-partition}]\hfill\\
\textsc{Instance:} Set $A$ with $|A|=3m$, a positive integer $p$, and a function $s \colon A \to \R2{\Z_{>0}}$ such that $p/4 < s(a) < p/2$ for each $a \in A$ and $\sum_{a \in A} s(a) = mp$. \\
\textsc{Question:} Is there a partition $\{A_i\}_{i \in [m]}$ of $A$ such that $\sum_{a \in A_i} s(a) = p$ for all $i \in [m]$?
\end{problem}

\begin{theorem}
    The decision versions of \textsc{min-max-bsf} and \textsc{max-min-bsf}  are $\NP$-complete on unweighted split graphs.
\end{theorem}
\proof{Proof.}
    We first provide a polynomial-time reduction from \textsc{3-partition} to the decision version of \emph{(weighted)} \textsc{min-max-bsf}. Then we modify the produced instance of \textsc{bsf} into an \emph{unweighted} instance.
    
    Let $(A,p,s)$ be an instance of \textsc{3-partition}.
    We create a graph $G$ on vertices $\{v_a, v'_a : a \in A\}$ and edges $\{\{v_a, v'_a\} : a \in A\} \cup \{\{v_a, v_b\} : a,b \in A \text{ with } a\neq b\}$.
    Clearly, $G$ is a split graph as $V':=\{v'_a \in V(G) : a \in A\}$ induces an independent set in $G$, and $V(G)\setminus V'$ induces a clique in $G$.
    {Moreover, each vertex $v'_a \in V'$ has a unique neighbor~$v_a$ in $V(G)\setminus V'$.}
    Let us define a weight function {$w$} such that $w(e)=s(a)$ if $e=\{v_a, v'_a\}$, and $w(e)=1$ otherwise.
    We next prove that $(A,p,s)$ is a \textsc{yes}-instance if and only if the edge-weighted graph $(G,w)$ contains a spanning $m$-forest $\mathcal{T}$  such that \(\max^{}_{T \in \mathcal{T}} w(T) \le p+2\).

    
    { Let $\mathcal{T}$ be a spanning $m$-forest such that \(\max^{}_{T \in \mathcal{T}} w(T) \le p+2\). 
If there exists a tree $T \in \mathcal{T}$ such that $|V(T) \cap V'| \leq 2$, then } {as $\mathcal{T}$ is a spanning $m$-forest of $G$ and $|V'|=|A|=3m$, there exists also a tree $T' \in \mathcal{T}$ such that $|V(T') \cap V'| \ge 4$. Then $|V(T')|\ge 8$ and $w(T') \geq \sum_{v'_{a} \in V(T')\cap V'} s(a) + 3$. Since $s(a)>p/4 $ for each $a \in A$, we have  $ \sum_{v'_{a} \in V(T')\cap V'} s(a)  > 4p/4 =p$. Putting all together, we have  $w(T') \geq \sum_{v'_{a} \in V(T')\cap V'} s(a) + 3 > 4p/4 +2 = p+2$.} {This \R2{contradicts} $\mathcal{T}$ being a spanning $m$-forest such that $\max^{}_{T \in \mathcal{T}} w(T) \le p+2$}.
    {Hence $|V(T) \cap V'| = 3$ for all $T \in \mathcal{T}$.
    As} $|A|=3m$ {and every vertex in $V'$ has a unique neighbor in $G$, it holds that each $T \in \mathcal{T}$} has exactly three leaves {in $V'$} and three internal vertices  {in $V(G)\setminus V'$}.
    {This well-defined topological structure of the trees in~$\mathcal{T}$ naturally induces a partition of $A$ which is defined as follows.}
    {Let} $A_T = \{a \in A : v_a \text{ is an internal vertex of } T\}$.
    {The construction of $(G,w)$ implies that every internal vertex~$v_a$ of $T$ is adjacent to a leaf~$\R2{v'_{a}}$ and~$w(\{v_{a},\R2{v'_{a}}\})=s(a)$, and so  $\sum_{a \in A_T} s(a) = \sum_{a \in A_T} w(\{v_a,\R2{v'_{a}}\}) = w(T)-2 \le p$} for all $T \in \mathcal{T}$.
    Because~$\sum_{a \in A} s(a) = mp$, we conclude that the collection $\{A_T\}_{T \in \mathcal{T}}$ forms a partition of $A$ into $m$ triplets such that $\sum_{a \in A_T} s(a)= p$ for all $T \in \mathcal{T}$. 
    To prove the converse, it suffices to observe that every partition $\{A_i\}_{i\in[m]}$ of $A$ with $\sum_{a \in A_i}s(a)=p$ for all $i \in [m]$ induces a spanning $m$-forest $\{T_i\}_{i \in [m]}$ such that $T_i$ is any tree in~$G$ on vertices $\bigcup_{a\in A_i} \{v_a,v'_a\}$ and $w(T_i)=p+2$.

    We now transform the weighted graph $(G,w)$ into an \emph{unweighted} instance as follows.
    For each vertex $v_a \in V(G)\setminus V'$,  add a set $L_a$ of $s(a)-1$ newly created vertices, and an edge $\{v_a,v\}$ for each $v \in L_a$.
    Let us denote by $G'$ the graph obtained using this transformation from~$G$, and note that $G'$ is also a split graph.
    Consider a spanning $m$-forest~$\mathcal{T}$ of $G'$ such that $w(T)=|E(T)|\leq p+2$ for all $T \in \mathcal{T}$.
    Note that, for each $a \in A$, all vertices in $L_a\cup\{v_a, v'_a\}$ belong to the same tree in $\mathcal{T}$ as $p\geq 3$.
    By the same arguments used for the weighted case,  we   conclude that 
     $(A,p,s)$ is a \textsc{yes}-instance if and only if $G'$ contains a spanning $m$-forest $\mathcal{T}$  such that \(\max_{T \in \mathcal{T}} w(T) \le p+2\).
    Since \textsc{3-partition} is \emph{strongly} $\NP$-complete{~\citep[see][]{garey1979computers}} and the decision version of \textsc{bsf} clearly belongs to $\NP$, it holds that the decision version of \textsc{Unweighted Balanced Spanning Forest} is $\NP$-complete.

    The reduction from \textsc{3-partition} to \textsc{max-min-bsf} is the same as described above. 
    Consider now a spanning $m$-forest~$\mathcal{T}$ of $G'$ such that $w(T)=|E(T)|\geq p+2$ for all $T \in \mathcal{T}$.
    If there exists a tree $T \in \mathcal{T}$ containing at most two vertices in $V(G)\setminus V'$, then $w(T) <p+2$, a contradiction.
    Using the same arguments in the proof for the min-max variant, it holds that the internal vertices of each of the trees in~$\mathcal{T}$ induce a partition $\{A_T\}_{T \in \mathcal{T}}$ of $A$ into $m$ triplets such that $\sum_{a \in A_T} s(a)= p$ for all $T \in \mathcal{T}$.
\Halmos \endproof

Using a reduction from the 3-dimensional matching problem, \citet{dyer1985complexity} showed that, for each fixed $k\ge2$, \textsc{max-min-bcp} is $\NP$-hard on unweighted bipartite graphs.
This immediately implies that \textsc{max-min-bsf} is also $\NP$-hard on unweighted bipartite graphs, for every fixed $k\ge 2$.
 By means of a reduction from \textsc{3-sat}, we next prove an analogous  result for the \textsc{min-max-bsf} and give an alternative proof for \textsc{max-min-bsf}.
Later on, we shall use this reduction from \textsc{3-sat} to establish a lower bound on the running time to solve theses problems under the Exponential-Time Hypothesis. 

\begin{problem}[\textsc{3-satisfiability} (\textsc{3-sat})]\hfill\\
\textsc{Instance:} 
{A boolean formula~$(X,\mathcal{C})$ in conjunctive normal form, where $X$ is the set of Boolean variables, and~$\mathcal{C}$ is a collection of clauses over $X$ such that each clause has exactly 3 literals.} \\
\textsc{Question:} Is there a truth assignment $\sigma \colon X \to \{\true,\false\}$ that satisfies every clause in~$\mathcal{C}$?
\end{problem}

\begin{theorem}\label{thm:hardness:3sat}
    The decision versions of \textsc{max-min-bsf} and \textsc{min-max-bsf}  are $\NP$-complete on unweighted bipartite graphs for each fixed~$k\ge2$.
\end{theorem}
\proof{Proof.}
    We first show a polynomial-time reduction from \textsc{3-sat} to \textsc{min-max-bsf} with $k=2$.
    Consider an instance~$\phi=(X,\mathcal{C})$ of \textsc{3-sat} containing a set~$X$ of $N$ variables and a set $\mathcal{C}$ of $M$ clauses.
    Let~$G$ be the graph obtained from~$(X, \mathcal{C})$ as follows.
    For each variable~$x \in X$, $G$ has two vertices~$v_x$ and~$v_{\bar x}$, a vertex $u_x$ which is adjacent to both $v_x$ and~$v_{\bar x}$, and  a set~$L_x$ containing exactly $2N+M-1$ vertices such that their single neighbor in $G$ is $u_x$.
    For each clause~$C \in \mathcal{C}$, the graph $G$ has a vertex~$z^{}_C$ which is adjacent to a vertex~$v_{y}$ if and only if~$y$ is a literal in~$C$.
    There are two additional vertices~$a$ and $b$, and each of them is adjacent to $v_x$ and $v_{\bar x}$ for all $x \in X$.
    Finally, $G$ contains a set~$L_a$ of $t-N-M-1$ vertices and a set~$L_b$ of $t-N-N(2N+M)-1$ vertices such that every vertex in $L_a$ (resp. $L_b$) is adjacent only to $a$ (resp. $b$) in $G$, where $t:=(N+1)(2N+M)$. Note that the size of $G$ is $\bigO(t) = \bigO(N^2+NM)$, that is, polynomial in the size of $\phi$.
    Moreover, it is clear from its construction that $G$ is bipartite.
    Figure~\ref{fig:reduction:3sat} depicts an example of this transformation.
    {Assume, for simplicity, that every edge of $G$ has unit weight, i.e., $w(e)=1$ for all $e \in E(G)$.}

    Suppose first that $G$ is spanned by a forest consisting of two trees $T_1$ and~$T_2$ such that $w(T_1)\leq w(T_2) \leq t-1$.
    Note that~$w(T_1)= w(T_2) = t-1$ because~$G$ has precisely $2t$ vertices.
    Note also that, for each vertex $v \in\{a,b\}$ or~$v$ associated with a variable in $X$, all vertices in~$L_v$ belong to the same tree containing~$v$ as $v$ is their only neighbor, and each of the trees $T_1$ and $T_2$ contains at least two vertices.  
    Since $a$ and $b$ are not adjacent in $G$, and $|L_a|+|L_b|=t-2$, any tree containing both $a$ and $b$ has weight larger than~$t-1$. 
    Hence, $a$ and $b$ belong to different trees in the spanning forest $\{T_1, T_2\}$.
    Let us assume (w.l.o.g.) that $a \in V(T_1)$ and $b \in V(T_2)$.
    Consider a variable $x \in X.$
    Suppose to the contrary that~$u_x \notin V(T_2)$, then $w(T_2) \leq 2t-1 - (|L_a|+1) - (|L_x|+1) \le 2t -1 - (t -N-M) - (2N+M) = t-1-N \leq t -2$, a contradiction.
    Thus, we have~$u_x \in V(T_2)$ for all $x \in X$.
    Since the set $\{v_x,v_{\bar x}\}$ is a separator of $u_x$ and $b$ in $G$, we have $|\{v_x,v_{\bar x}\} \cap V(T_1)|\le 1$ for every $x \in X$.
    If $\{v_x, v_{\bar x}\} \cap V(T_1)=\emptyset$, then $w(T_1)\leq |L_a|+N-1+M=t-2$, a contradiction.
    Hence, it holds that, for every $x \in X$,  $T_1$ {contains either~$v_x$ or~$v_{\bar x}$.} 
    Finally, one may easily verify that  $z^{}_C \in V(T_1)$ for all $C \in \mathcal{C}$.
    We now define an assignment $\sigma \colon X \to \{\true, \false\}$ such that $\sigma(x)=\true$ if and only if $v_x \in V(T_1)$.
    It follows from the construction of~$G$ that~$\sigma$ satisfies~$\phi$.

   Let~$\sigma \colon X \to \{\true, \false\}$ be an assignment that satisfies the Boolean formula~$\phi$.
   We construct the \R2{following partition}~$\{V_1,V_2\}$ of $V(G)$:
   \begin{align*}
       V_1:= & \{a\}\cup L_a \cup\{v_x \colon x\in X \text{ and } \sigma(x) = \true\} \cup\{v_{\bar x} \colon x\in X \text{ and } \sigma(x) = \false\}\\
       & \cup\{z_C \colon C\in \mathcal{C}\}; \:\:\text{and}\\
       V_2:= &\{b\}\cup L_b \cup\{v_x \colon x\in X \text{ and } \sigma(x)=\false\}\cup\{v_{\bar x} \colon x\in X \text{ and } \sigma(x)=\true\}\\
       & \cup \left(\bigcup_{x\in X} \R2{\left(L_x \cup \{u_x\}\right)}\right).
   \end{align*}
   Since~$\sigma$ satisfies~$\phi$, for each clause $C \in \mathcal{C}$, there exists a variable $x \in X$ such that either~$\sigma(x)=\true$ and~$x \in C$, or~$\sigma(x)=\false$ and~$\bar{x} \in C$.
    As a consequence, $V_1$ induces a connected subgraph of $G$.
    Observe that $|V_1| = |V_2 |= t$, and that there exist trees $T_1$ and $T_2$ spanning $V_1$ and $V_2$, respectively.
    Thus, we have $w(T_1) = w(T_2) = t-1$.
    It follows from the previous discussion that $\phi$ is a \textsc{yes}-instance of \textsc{3-sat} if and only if $G$ contains a spanning forest with two trees having weight at most $t-1$.

    The reduction from \textsc{3-sat} to \textsc{max-min-bsf} is the same as previously described.
    If $G$ has a spanning 2-forest $\{T_1, T_2\}$ such that $w(T_1)\ge w(T_2)\geq t-1$, then $w(T_1)=w(T_2)=t-1$ since $G$ has $2t$ vertices.
    Therefore, using the same reasoning for the min-max case, we conclude that $\phi$ is a \textsc{yes}-instance of \textsc{3-sat} if and only if $G$ contains a spanning forest with two trees having weight at least~$t-1$. 

    Suppose now that $k>2$.
    First construct the graph~$G$ as before. 
    Then add to it $k-2$ paths, say~$P_3, P_4, \ldots, P_k$, each of them formed by $t$ newly created vertices, and link vertex~$a$ in $G$ to a single endpoint of each of these paths.
    Let $G'$ be the graph obtained from $G$ by this transformation, and note that $G'$ has exactly $kt$ vertices.
    If $G'$ admits a spanning $k$-forest~$\{T_i\}_{i\in [k]}$ with $\max_{i \in [k]} w(T_i)\leq t-1$ or $\min_{i \in [k]} w(T_i)\geq t-1$, then $w(T_i)=t-1$ for all $i \in [k]$.
    Thus, for every $i \in \{3,4, \ldots, k\}$,  $P_{i}$ is a tree in this forest.
    Assuming without loss of generality that $P_{i}=T_i$ for all $i \in \{3,4, \ldots, k\}$, the previous discussions for the case $k=2$ also imply that, for each fixed $k\ge 2$, $\phi$ is a \textsc{yes}-instance of \textsc{3-sat} if and only if \R2{$G'$} contains a spanning forest with $k$ trees having weight exactly $t-1$.

    {The result follows since \textsc{3-sat} is $\NP$-complete \citep[see][]{garey1979computers}.}
\Halmos \endproof

\begin{figure}[t!]
\centering
    \centering
    \begin{tikzpicture}[
            auto,
            node distance = 1.3cm, 
            semithick, 
            scale=0.6, 
            every node/.style={scale=0.75}
        ]

        \tikzstyle{every state}=[
            draw = black,
            semithick,
            fill = white,
            minimum size = 5mm,
            inner sep=1pt
        ]

        \node[state] (xu) []{$u_x$};
        \node[state] (x) [above left of=xu] {$v_x$};
        \node[state] (xn) [below left of=xu] {$v_{\bar x}$};
        \node (xl1) [left=1cm of xu, shape=circle,scale=0.7, inner sep=0pt]{\rotatebox{90}{$\cdots$}};
        \node (xl2) [above=1mm of xl1, draw=black, shape=circle, inner sep=2pt]{};
        \node (xl3) [below=1mm of xl1, draw=black, shape=circle, inner sep=2pt]{};
        \node[draw,dotted, inner sep=6pt, fit=(xl1) (xl2) (xl3), label={[label distance=-1pt]180:$L_x$}] {};
        
        \node[state] (c1) [left=2cm of x] {$z_{1}$};
        \node [left=0.2cm] at (c1) {$(x\lor \bar y \lor z)$};
        \node[state] (c2) [left=2cm of xn] {$z_{2}$};
        \node [left=0.2cm] at (c2) {$(\bar x\lor \bar y \lor z)$};
        
        \node[state] (a) [above right=3cm of xu] {$a$};
         \node (al1) [right=1cm of a, shape=circle,scale=0.7, inner sep=0pt]{\rotatebox{90}{$\cdots$}};
        \node (al2) [above=1mm of al1, draw=black, shape=circle, inner sep=2pt]{};
        \node (al3) [below=1mm of al1, draw=black, shape=circle, inner sep=2pt]{};
        \node[draw,dotted, inner sep=6pt, fit=(al1) (al2) (al3), label={[label distance=-1pt]0:$L_a$}] {};
        
        \node[state] (b) [below right=3cm of xu] {$b$};
        \node (bl1) [right=1cm of b, shape=circle,scale=0.7, inner sep=0pt]{\rotatebox{90}{$\cdots$}};
        \node (bl2) [above=1mm of bl1, draw=black, shape=circle, inner sep=2pt]{};
        \node (bl3) [below=1mm of bl1, draw=black, shape=circle, inner sep=2pt]{};
        \node[draw,dotted, inner sep=6pt, fit=(bl1) (bl2) (bl3), label={[label distance=-1pt]0:$L_b$}] {};

        \node[state] (yn) [above of=x] {$v_{\bar y}$};
        \node[state] (yu) [above right of =yn]{$u_y$};
        \node[state] (y) [above left of=yu] {$v_y$};
        \node (yl1) [left=1cm of yu, shape=circle,scale=0.7, inner sep=0pt]{\rotatebox{90}{$\cdots$}};
        \node (yl2) [above=1mm of yl1, draw=black, shape=circle, inner sep=2pt]{};
        \node (yl3) [below=1mm of yl1, draw=black, shape=circle, inner sep=2pt]{};
        \node[draw,dotted, inner sep=6pt, fit=(yl1) (yl2) (yl3), label={[label distance=-1pt]180:$L_y$}] {};

        \node[state] (z) [below of=xn] {$v_{z}$};
        \node[state] (zu) [below right of =z]{$u_z$};
        \node[state] (zn) [below left of=zu] {$v_{\bar z}$};
        \node (zl1) [left=1cm of zu, shape=circle,scale=0.7, inner sep=0pt]{\rotatebox{90}{$\cdots$}};
        \node (zl2) [above=1mm of zl1, draw=black, shape=circle, inner sep=2pt]{};
        \node (zl3) [below=1mm of zl1, draw=black, shape=circle, inner sep=2pt]{};
        \node[draw,dotted, inner sep=6pt, fit=(zl1) (zl2) (zl3), label={[label distance=-1pt]180:$L_z$}] {};

        \path (c1) edge node{} (x);
        \path (c1) edge node{} (yn);
        \path (c1) edge node{} (z);

        \path (c2) edge node{} (xn);
        \path (c2) edge node{} (yn);
        \path (c2) edge node{} (z);

        \path (yu) edge node{} (y);
        \path (yu) edge node{} (yn);
        \path (yu) edge node{} (yl2);
        \path (yu) edge node{} (yl3);

        \path (xu) edge node{} (x);
        \path (xu) edge node{} (xn);
        \path (xu) edge node{} (xl2);
        \path (xu) edge node{} (xl3);

        \path (zu) edge node{} (z);
        \path (zu) edge node{} (zn);
        \path (zu) edge node{} (zl2);
        \path (zu) edge node{} (zl3);

        \path (a) edge node{} (y);
        \path (a) edge node{} (yn);
        \path (a) edge node{} (x);
        \path (a) edge[bend left=20] node{} (xn);
        \path (a) edge[bend left=25] node{} (z);
        \path (a) edge[bend left=35] node{} (zn);
        \path (a) edge node{} (al2);
        \path (a) edge node{} (al3);
        
        \path (b) edge[bend right=35] node{} (y);
        \path (b) edge[bend right=25] node{} (yn);
        \path (b) edge[bend right=20] node{} (x);
        \path (b) edge node{} (xn);
        \path (b) edge node{} (z);
        \path (b) edge node{} (zn);
        \path (b) edge node{} (bl2);
        \path (b) edge node{} (bl3);
    \end{tikzpicture}%
    \caption{Example of the transformation described in the proof of Theorem~\ref{thm:hardness:3sat} on input $(x\lor \bar y \lor z) \land (\bar x\lor \bar y \lor z)$.}
    \label{fig:reduction:3sat}
\end{figure}

In what follows, we show that the existence of a subexponential-time algorithm for \textsc{bsf} is very unlikely as such an algorithm would imply a breakthrough in satisfiability solving.
This is obtained under the \emph{Exponential-Time Hypothesis} introduced by~\cite{ImpPat01}.
The general strategy is to apply the \emph{Sparsification Lemma} 
to transform \textsc{3-sat} instances into an equivalent set of instances whose number of clauses is linear in the number of variables, and then use the reduction presented in Lemma~\ref{lemma:reduction}.

\begin{conjecture}[\textit{Exponential-Time Hypothesis (ETH)}, \citet{ImpPat01}]
There exists \R2{a} $\lambda > 0$ such that \textsc{3-sat} cannot be solved in $\bigO(2^{\lambda N})$ time on
$N$-variable instances.
\end{conjecture}

\begin{lemma}[\textit{Sparsification Lemma}, \citet{ImpPatZan01}]\label{thm:sparsification}
Let $\gamma > 0 $, and let~$\phi$ be a \textsc{3-sat} instance with $N$ variables. 
There is a number $c = c(\gamma)$, and an algorithm that computes \textsc{3-sat} instances $\phi_1, \ldots, \phi_\ell$ from $\phi$ with~$\ell \leq 2^{\gamma N}$ in~$\bigO(2^{\gamma N}\poly(N))$ time such that
\begin{enumerate}[(i)]
	\item $\phi_i$ has $N$ variables and at most~$cN$ clauses for every~$ i \in [\ell]$; and
 	\item $\phi$ is satisfiable if and only if $\phi_i$ is satisfiable for some $i \in [\ell]$.
\end{enumerate}
\end{lemma}

\begin{lemma}\label{lemma:reduction}
    Let $k\ge 2$ be an integer.
    Let~$\phi$ be an instance of~\textsc{3-sat} with~$N$ variables and~$M$ clauses.
	 There exists an algorithm that computes a weighted graph~$(G,w)$ from  $\phi$ in $\bigO(N+M)$ time such that 
    \begin{enumerate}[(i)]
        \item $G$ is bipartite and contains $\bigO(N+M)$ vertices and edges;
        \item $(G,w)$ has a spanning $k$-forest $\{T_i\}_{i\in [k]}$ with $w(T_i)=(N+1)(2N+M)-1$ for every $i \in [k]$ if and only if $\phi$ is satisfiable.
    \end{enumerate}
\end{lemma}
\proof{Proof.}
    The reduction described in the proof of Theorem~\ref{thm:hardness:3sat} builds an unweighted graph $G'$ from~$\phi=(X,\mathcal{C})$ where each of the vertices in $U:=\{a,b\}\cup\{u_x : x \in X\}$ is adjacent to a potentially large (yet polynomial on the size of $\phi$) number of vertices of degree~1.
    For each $v \in U$, we replace all vertices in $L_v$ by a single vertex $p_v$, and add an edge $\{v,p_v\}$ of weight $|L_v|$.
    Additionally, each path $P_i$ with $i \in \{3,\ldots, k\}$ is replaced by two newly created vertices $p_i,q_i$, and edges $\{p_i,q_i\}$ and $\{p_i,a\}$  with weights~$(N+1)(2N+M)-1$ and $1$, respectively.
    Let $G$ be the graph obtained from this transformation on $G'$, and observe that $G$ contains $4N+M+4 + 2(k-2)$ vertices, 
    $7N+3M+2 + 2(k-2)$ edges, 
    and it is bipartite as $G'$ is also bipartite.
    Hence, this algorithm takes $\bigO(N+M)$ time since $k$ is constant.
    It follows from the proof of Theorem~\ref{thm:hardness:3sat} that $(G,w)$ has a spanning $k$-forest $\{T_i\}_{i \in [k]}$ with $w(T_i)=(N+1)(2N+M)-1$ if and only if $\phi$ is satisfiable.
\Halmos \endproof

\begin{theorem}\label{thm:lower_bound}
	For each integer~$k\ge2$,
    there exists an $\varepsilon> 0$ such that neither \textsc{max-min-bsf} nor \textsc{min-max-bsf}  can be solved in  $\bigO(2^{\varepsilon n})$ time on $n$-vertex bipartite graphs, unless ETH fails.
\end{theorem}
\proof{Proof.}
    {We next show how to design an algorithm that solves \textsc{3-sat} in subexponential time, assuming the existence of a subexponential-time algorithm for either \textsc{max-min-bsf} or \textsc{min-max-bsf}, which would imply that ETH fails.}
    
    Let~$k\ge2$ be an integer.
    Let $\phi$ be a \textsc{3-sat} instance with $N>2$ variables.
    Suppose to the contrary that either \textsc{max-min-bsf} or \textsc{min-max-bsf} restricted to instances consisting of $n$-vertex bipartite graphs  (with $k$ fixed) can be solved in  $\bigO(2^{\varepsilon n})$ time for every $\varepsilon>0$.
    Let $\varepsilon>0$, define $\gamma = \varepsilon/2$, and apply the Sparsification Lemma on $\phi$.
    Let  $c = c(\gamma)$ be a constant, and let~$\phi_1,\ldots, \phi_\ell$ with $\ell \leq 2^{\gamma N}$ be the \textsc{3-sat} instances computed from $\phi$ in~$\bigO(2^{\gamma N}\poly(N))$ time  as in the statement of Lemma~\ref{thm:sparsification}.
    For every~$i \in [\ell]$,  let $M_i$ denote the number of clauses of~$\phi_i$. 
    
    For each~$i \in [\ell]$,  construct
  a weighted graph~$(G_i,w_i)$ from~$\phi_i$ as in Lemma~\ref{lemma:reduction}, and define $t_i=(N+1)(2N+M_i)-1$. 
  Since $\phi_i$ has $N$ variables and at most $cN$ clauses, the size of $G_i$ is $\bigO(N)$.
  Let $\rho= {\varepsilon}/({2(4+c)})$, and let~$\mathcal{A}$ be an algorithm that solves one of the two versions of \textsc{bsf} in~$\bigO(2^{\rho n})$ time, where $n$ is the number of vertices of the input graph. 
  Hence, the \R2{running time} of~$\mathcal{A}$ on input~$(G_i, w_i, k, t_i)$ is
	\(\bigO(2^{\rho |V(G_i)|}) = \bigO(2^{\rho (4+c)N}) = \bigO(2^{\eps N/2})\), where the first equation holds since the number of vertices of  $G_i$ is $\bigO(M_i+4N)$ as in the proof of Lemma~\ref{lemma:reduction}.
 By Lemmas~\ref{thm:sparsification} and~\ref{lemma:reduction},  the satisfiability of $\phi$ can be decided by running~$\mathcal{A}$ on input~$(G_i, w_i, k, t_i)$ for every~$i \in [\ell]$. 
 This algorithm takes time
        \(\bigO(\ell 2^{\eps N/2}) = \bigO(2^{\gamma N} 2^{\eps N/2}) =\bigO(2^{\eps N}).\)
	Therefore, one can solve \textsc{3-sat} in $\bigO(2^{\eps N})$ time on $N$-variable instances for
    every~$\eps >0$, a contradiction to ETH.
\Halmos \endproof


\section{A $k$-approximation algorithm and heuristics }
\label{sec:heur}


{We propose  a simple (although tight) approximation algorithm for~\textsc{min-max-bsf}: compute a \R2{minimum-weight} spanning tree, then construct a spanning $k$-forest by removing {the} $k-1$ {heaviest} edges of this tree.}
To the best of our knowledge, this is the first approximation for \textsc{min-max-bsf} in the literature.
This approximation is then used to design heuristics that provide initial solutions to the solving methods proposed in the following sections.

\label{sec:approximation}
\begin{proposition}\label{prop:approx}
    \textsc{min-max-bsf} admits a tight $k$-approximation algorithm, which runs in $\bigO(m \log n)$ time on graphs with $n$ vertices and $m$ edges.
\end{proposition}
\proof{Proof.}
    Let $(G,w,k)$ be an instance of \textsc{min-max-bsf}, where $G$ is a graph with $n$ vertices and $m$ edges.
    First, the algorithm computes a minimum-weight spanning tree $T$ of $(G,w)$ in $\bigO(m \log n)$ time using Kruskal's algorithm.
    Since this algorithm sorts the edges by their weights in a \R2{non-decreasing} fashion, one may easily obtain a set $E'$ containing the $k-1$ heaviest edges in $T$.
    Finally, the algorithm outputs the spanning forest, say~$F$, obtained from $T$ by removing the edges in $E'$.
    Obviously, the weight of the heaviest tree in $F$ is at most $w(T)- \sum_{e \in E'} w(e)$.

    Consider now an optimal spanning $k$-forest~$F^*$ of $(G,w,k)$ such that its heaviest tree has weight equal to $\opt$.
    Let $E^*$ be a subset of $k-1$ edges of the spanning tree $T$ such that~$F^*$ plus these edges yield a spanning tree~$H$ of $G$.
    Clearly, it holds that $w(T) \leq w(H) \leq k\cdot \opt + \sum_{e \in E^*} w(e)$.
    Hence, to conclude that the algorithm described above is a $k$-approximation, it suffices to observe that $\sum_{e \in E'} w(e) - \sum_{e \in E^*} w(e) \ge 0$. 
    Therefore, the weight of the heaviest tree in $F$ is at most~$k \cdot \opt$.

The analysis of the greedy algorithm described above is tight as shown by the following family of \emph{unweighted} instances.
For every integer $k\ge 2$, define a graph $G$ which is obtained from the disjoint union of~$k$ paths~$\{P_i\}_{i \in [k]}$ of length~$2$, and $k-1$ edges $\{u_k, u_i\}$ with $i \in [k-1]$, where $u_i$ is an arbitrary (but fixed) endpoint of $P_i$ for all $i \in [k]$.
Observe that $G$ is a tree having exactly $k$ leaves and $3k-1$ edges, and that the greedy algorithm can choose to remove precisely $k-1$ of the edges incident with {the leaves}.
In this case, the solution produced by the greedy algorithm has value $2k$.
On the other hand, the optimal value for the \textsc{min-max-bsf} instance $(G,k)$ is $2$. 
\Halmos \endproof



It seems that selecting the edges to remove from the minimum spanning tree more carefully in the proof of Proposition~\ref{prop:approx} could yield \R2{better-quality}  solutions. 
To this end, {we use the algorithm devised by~\citet{vaishali2018}, which solves \textsc{min-max-bsf} on trees in $\bigO(kn^3)$ time.}
{This idea yields our first heuristic: compute a minimum spanning tree of $G$, and return the solution obtained by applying the algorithm of \citeauthor{vaishali2018} to that tree.}
{This heuristic coincides with \Cref{alg:heur} when the loop in line~\ref{line:while} is executed only once.} 
The following family of instances demonstrates that this heuristic may lead to solutions with approximation ratios arbitrarily close to $k/2$.

Let $k\geq 4$ be an even integer number, and let~$\tau \geq 2$ \R2{be integer}. 
We construct a graph $G$ by taking the union of~$k$ disjoint  paths~$\{P_i\}_{i \in [k]}$ of length~$3$, where $u_i$ and $v_i$ are the endpoints of path $P_i$ for $i\in [k]$, $k-1$ edges $\{u_k, u_i\}$ with $i \in [k-1]$ and $k-1$ edges $\{v_k, v_i\}$ with $i \in [k-1].$ The edges on the paths~$\{P_i\}_{i \in [k]}$ except the middle ones have weight 1 and these middle edges have weight $\tau$. 
The edges $\{u_k, u_i\}$ and  $\{v_k, v_i\}$ with $i \in [k-1]$  have weight $\tau$.
Let~$w$ denote the weight of the edges of $G$ as defined above.

Removing the middle edges on $P_i$ for $i\in [k-1]$ of weight $\tau$ produces a minimum-weight spanning tree, say~$T$. 
The following forest is an optimal solution to the \textsc{min-max-bsf} instance $(T,k,w)$: $k-2$ trees of weight 1 and two trees, one including $u_k$ and the other including $v_k$, of weight $k(\tau+1)/2 +1$. 
Note that the optimal value of~$(G,k,w)$ is precisely
\(\tau+2\).
Hence, optimally solving \textsc{min-max-bsf} on a minimum spanning tree of the input graph $G$ produces a solution of value tending to~$\frac{k}{2}\opt$ as $\tau$ increases.

Our second heuristic applies the first heuristic on different initial graphs that are obtained by removing some edges. 
The steps are detailed in Algorithm~\ref{alg:heur}. Starting from the initial graph with no forbidden edges, we iteratively compute a minimum spanning tree \( T_u \) excluding the forbidden set \( H_u \). If the \R2{heuristic score} \( \R2{\mathrm{HS}_u = \left\lceil \frac{w(T_u)}{k} \right\rceil} \) is less than the current global upper bound \( \mathrm{UB} \), we use \citeauthor{vaishali2018}'s algorithm to find a minimum spanning (MSKF) \( k \)-forest \( F_u \) for $T_u$, from which an upper bound \( \mathrm{UB}_u \) is computed. 
\R2{This score is not generally a valid lower bound, and pruning based on it may discard nodes containing improving solutions.}
The algorithm proceeds by generating child nodes where each edge \( e \in E(T_u) \) is added to the forbidden set. 


\begin{algorithm}[t]
\caption{Heuristic for \textsc{min-max-bsf}}
\label{alg:heur}
\begin{algorithmic}[1]

\Require{{An instance $(G,k,w)$ of \textsc{min-max-bsf}.}}
\Ensure{{A spanning $k$-forest of $G$.}}
\State Initialize $\mathrm{UB} \gets \infty$, {$F \gets \emptyset$}, node number \( u = 0 \), and  forbidden edge set for node $0$, \( H_0 \gets \emptyset \)
\State Compute \( G_{0} \gets G 
\)
\State Compute \( T_{0} \gets \text{MST}(G_{0}) \) and \(\R2{\mathrm{HS}_{0} = \left\lceil \frac{w(T_0)}{k} \right\rceil}\)         \Comment{Kruskal's algorithm}
\State Insert node \( u = 0 \) into queue \( \mathcal{Q} \)
\While{\( \mathcal{Q} \neq \emptyset \) and {the time limit} {is not reached}}\label{line:while}
    \State Pop node \( u \) from \( \mathcal{Q} \) with smallest \(\R2{\mathrm{HS}_{u}}\)    
    \If{\( \R2{\mathrm{HS}_u} < \mathrm{UB} \)}
    \State Let  \( F_u =\{F_{u,i}\}_{i\in [k]} \gets \text{MSKF}(T_u)\)
    {\Comment{Algorithm for \textsc{min-max-bsf} on trees due to \cite{vaishali2018}}} 
    \State Let \( \text{UB}_u \gets \max_{i\in[k]} w(F_{u,i}) \)
    \If{\( \text{UB}_u < \text{UB} \)}
         \(\text{UB} \gets \text{UB}_u\), \:  {\(F\gets F_u\)}
    \EndIf
    \ForAll{edges \( e \in T_u \)}
        \State Create child node \( u' \) with \( H_{u'} \gets H_u \cup \{e\} \)
        \State Compute \( G_{u'} \gets G \setminus H_{u'} \)
        \State Compute \( T_{u'} \gets \text{MST}(G_{u'}) \)
        \If{\(T_{u'}\) exists and \( \R2{\mathrm{HS}_{u'} = \left\lceil \frac{w(T_{u'})}{k} \right\rceil} < \mathrm{UB} \)}
            \State Insert node \( u' \) into \( \mathcal{Q} \)
        \EndIf        
    \EndFor
    \EndIf
\EndWhile
\Return {$F$}
\end{algorithmic}
\end{algorithm}

{We shall see in Section~\ref{sec:experiments} that this heuristic often provides very good solutions.}
Therefore, it is used to produce primal solutions for each of the solving methods proposed in this work.

The strategy used in the proposed approximation for \textsc{min-max-bsf} does not seem to work for the max-min variant of the problem.
Indeed, we are not aware of any approximation for \textsc{max-min-bsf}.

\section{Integer {linear programming} formulations}
\label{sec:formulations}

We next show three integer {linear programming} formulations for both the max-min and the min-max variants of \textsc{bsf}.
The first one uses flows, and contains a polynomial number of variables and constraints.
The second model has a polynomial number of variables and uses cycle-elimination constraints.
The third one contains a binary variable for every possible tree in the input graph.

\subsection{Flow formulations} \label{sec:form:flow}
We first present a compact formulation for \textsc{min-max-bsf} that models a solution as a collection of arborescences -- rooted trees where every \R2{arc} is oriented away from the root -- in an auxiliary directed graph. 
Let $(G=(V,E),w,k)$ be an instance of this problem \R2{with $w(e) >0$ for all $e \in E$}.
We define $A=\{(u,v),(v,u): \{u,v\}\in E\}$ to be the set of all arcs corresponding to edges in $E$.
{We use $w(u,v)=w(v,u)=w(\{u,v\})$ for all $\{u,v\}\in E$. }

For each $v \in V$, the formulation has binary variable $y_v$ which represents if $v$ is the root of an arborescence in the spanning \R2{$k$-forest}, and a \R2{non-negative} variable $g_v$ to be the weight of the arborescence rooted at $v$ if $v$ is a root, and zero otherwise.
For each $a \in A$,  the model contains a binary variable $z_a$ that indicates if arc $a$ is part of the forest and a \R2{non-negative} variable $f_a$ that represents the amount of flow on arc $a\in A$. 
The formulation also has a \R2{non-negative} variable $\omega$ which corresponds to an upper bound on the weight of the heaviest tree in any solution.  

For $S\subseteq V$, we define  $\delta^+(S)=\{(u,v)\in A: u\in S, v\in V\setminus S\}$, $\delta^-(S)=\{(u,v)\in A: u\in V\setminus S, v\in S\}$ and $A(S)=\{(u,v)\in A: u\in S, v\in S\}$ in the graph $(V,A)$. If $S=\{v\}$, we use $\delta^-(v)$ and $\delta^+(v)$ instead of $\delta^-(\{v\})$ and $\delta^+(\{v\})$.

Let $\UB \in \Z_{\geq 0}$ be an upper bound on the value of an optimal solution to $(G,w,k)$. The flow model for \textsc{min-max-bsf} on this instance is as follows. 
\begin{align} 
        \min \: &  \omega &  \label{o} \\
        \text{s.t.} \: & \omega \geq g_v     & \forall v\in V\label{s1},\\
& \sum_{v\in V} y_v =k, & \label{s2} \\
        & \sum_{a \in \delta^-(v)} z_a +y_v =1 & \forall v \in V, \label{s3}\\
        & g_v+\sum_{a \in \delta^-(v)} f_a - \sum_{a \in \delta^+(v)} f_a = \sum_{a \in \delta^-(v)} w(a) z_a & \forall v \in V, \label{s5}\\
        &w(a) z_{a} \leq  f_{a}\leq {\mathrm{UB}}  z_{a} & \forall a \in A, \label{s6}\\
        &0 \leq  g_v\leq {\mathrm{UB}}  y_v & \forall v \in V, \label{s6p}\\
        & z_{a} \in \{0,1\} & \forall a \in A, \label{s8}\\
        & y_v \in \{0,1\} & \forall v \in V, \label{s9} \\
        & \omega  \ge 0. & \label{s10} 
\end{align}
Constraints \eqref{s2} ensure that a solution has exactly~$k$ vertices as roots.
Constraints \eqref{s3} guarantee that every vertex $v\in V$ is either a root or has exactly one incoming arc. 
Constraints \eqref{s6} bound~$f_a$ between $w(a)$ and $\UB$ if arc $a\in A$ is part of the forest, and force $f_a$ to be 0 otherwise. 
Similarly, constraints \eqref{s6p} force $g_v$ to be 0 if vertex $v\in V$ is not a root. 
The \R2{flow-conservation} constraints \eqref{s5}  model the weights of trees as follows.
If a vertex $v\in V$ is not a root, then the net flow into $v$ (incoming flow minus outgoing flow) is equal to the weight of the single incoming arc of~$v$. 
If~$v$ is a root, then $v$ has no incoming arc and $g_v$ is equal to sum of the flows on the outgoing arcs of~$v$. 
The sum of these flows is equal to the weight of the tree rooted at~$v$. 
The objective function \eqref{o} and constraints \eqref{s1} ensure that $\omega$ equals the weight of the heaviest tree. 

We strengthen the above model with the following inequalities: 
    \(\omega \geq \frac{1}{k}\sum_{v\in V}g_v\),
and $z_{uv}+z_{vu}\leq 1$ for all $\{u,v\}\in E$. {To see how the first inequality may strengthen the linear programming relaxation, consider an instance where we are given a path on vertices $\{1,2,3,4\}$ and edges $\{\{1,2\},\{2,3\},\{3,4\}\}$ with unit weights, and let $k=2$. 
The optimal value is 1, and suppose that we have $\UB=1$. The following solution is feasible for the LP relaxation: $y_v=g_v=1/2$ for $v=1, \ldots, 4$, $z_{12}=z_{21}=z_{34}=z_{43}=f_{12}=f_{21}=f_{34}=f_{43}=1/2$, $\omega=1/2$ and all other variables are zero. 
This fractional solution violates the first proposed inequality which implies $\omega \ge 1$.
}


Assuming an order~$\prec$ on the vertices of $G$, we also include the following 
inequalities 
\begin{align*}
    y_v + z_{vu} \leq 1 \quad (v,u)\in A: u\prec v
\end{align*}
to eliminate some symmetries by not allowing to use arcs outgoing from a root 
to lower indexed vertices. 


Regarding \textsc{max-min-bsf}, one can model it as follows. 
\begin{align} 
        \max \: &  \omega &  \label{op} \\
        \text{s.t.} \: &\text{\eqref{s2} -- \eqref{s5}, \eqref{s8}--\eqref{s10},} \notag\\
        & \omega \leq g_v +\UB(1-y_v)     & \forall v\in V,\label{s1p}\\
                &w(a) z_{a} \leq  f_{a}\leq Wz_{a} & \forall a \in A, \label{2s6}\\
        &0 \leq  g_v\leq W  y_v & \forall v \in V, 
\end{align}
 where $\UB$ is an upper bound for \textsc{max-min-bsf} and \R2{$\mathrm{W}$ is the weight of a maximum-weight spanning tree in~$(G,w)$.}

 \R2{In the previous formulations, we assumed that all edge weights are positive. If the input weighted graph contains non-positive cycles, we also need to introduce cycle-elimination constraints, as in the formulation presented next. }

\subsection{\R2{Cycle-elimination} formulations}\label{sec:form:elimination}

We next show a natural assignment formulation for \textsc{min-max-bsf}, which only uses binary edge variables to represent a spanning forest.
Given an instance $(G,w,k)$,  this model contains a binary variable $x_{ei}$ for every $e \in E(G)$ and $i \in [k]$ indicating that edge~$e$ belongs to the $i$-th tree in the forest.
Before presenting our model, let us introduce the following notation.
For every set of nodes $S \subseteq V$, we define $E(S) = \{\{u,v\} \in E : u,v \in S\}$, and~$\delta(S)=\{ \{u,v\} \in E :  u \in S, v \notin S\}$. The formulation is as follows.
\begin{align} 
        \min \:\: &  \omega & \label{cut:obj}  \\
        \text{s.t.} \:\:&  \omega \ge \sum_{e \in E} w(e)x_{ei} & \forall i \in [k], \label{ineq:ordering}\\
        & \sum_{i \in [k]} x_{ei} \le 1 & \forall e\in E, \label{atmost1} \\  
        &  x_{ei} + \sum_{j\in[k]\setminus\{i\}} x_{fj} \le 1 & \forall v \in V, e,f\in \delta(v), i \in [k],\label{ineq:disjoint1}\\
        & \sum_{i \in [k]} \sum_{e \in E(S)} x_{ei}\leq |S|-1 & \forall S\subseteq V \text{ with } S\neq \emptyset, \label{ineq:subtour}\\
        & \sum_{i\in [k]} \sum_{e\in E}  x_{ei} = n-k, &  \label{ineq:edgenumber}\\
        & x_{ei} \in \{0,1\} & \forall e \in E, i \in [k]. \label{var:x}
\end{align}

Constraints~\eqref{ineq:subtour} and~\eqref{ineq:edgenumber} ensure that any solution corresponds to a spanning $k$-forest of $G$.
Constraints~\eqref{atmost1} and~\eqref{ineq:disjoint1}  ensure that every edge belongs to at most one tree and that the trees are pairwise disjoint, respectively.
The objective function together with~\eqref{ineq:ordering} implies that $\omega$ equals the weight of the heaviest tree in any optimal solution.
A similar formulation for \textsc{max-min-bsf} can be easily obtained by changing~\eqref{ineq:ordering} as follows:
\begin{align} 
        \max \:\: &  \omega &  \nonumber\\
        \text{s.t.} \:\:& \omega \le \sum_{e \in E} w(e)x_{ei} & \forall i \in [k],   \\
        & \eqref{atmost1}-\eqref{var:x}. \nonumber 
\end{align}

The separation problem associated with inequalities~\eqref{ineq:subtour} can be solved as follows (see \citet{magnanti1995optimal}, pg.~48). 
Let $\bar x\in \R^{|E|k}$ be a feasible solution of the linear relaxation of the \R2{cycle-elimination} formulation.
The separation problem of the \R2{cycle-elimination} inequalities is equivalent to finding a \R2{non-empty} set of vertices $S \subseteq V$ that minimizes $|S|+\sum_{e \in E\setminus  E(S)} \sum_{i \in [k]} \bar x_{e,i}$.
Let $D=(V\cup\{s,t\}, A')$ be a directed graph obtained from the input graph $G=(V,E)$ such that $A':=\{(u,v), (v,u) : \{u,v\} \in E\}\cup\{(s,u) : u \in V\}\cup\{(u,t) : u \in V\}$.
Additionally, we define capacities $c \colon A' \to \R2{\R_{\geq 0}}$ as $c(u,v) = \sum_{i \in [k]} \bar x_{\{u,v\},i}/2$ if $u,v \in V$, $c(u,v) = \sum_{e \in \delta(v)} \sum_{i \in [k]} \bar x_{e,i}/2$ if $u=s$, and $c(u,v) = 1$ if $v=t$.

Suppose that $S\cup \{s\}$ is an $s,t$-cut in $D$ with $S\subseteq V$.
The capacity of this cut is
\begin{align*}
  c(S\cup\{s\}) = & |S|+\frac{1}{2}\sum_{v \in V\setminus S} \sum_{e \in \delta(v)} \sum_{i \in [k]} \bar x_{e,i} + \frac{1}{2} \sum_{e \in \delta(S)} \sum_{i \in [k]} \bar x_{e,i}  \\
   = &|S|+\sum_{e \in  E(V\setminus S)} \sum_{i \in [k]} \bar x_{e,i} + \sum_{e \in \delta(S)} \sum_{i \in [k]} \bar x_{e,i} \\
   = &|S|+\sum_{e \in E \setminus  E(S)} \sum_{i \in [k]} \bar x_{e,i}.
\end{align*}
To avoid $S$ from being empty, we solve the \R2{minimum-cut} problem for each $v \in V$ by setting $c(s,v)$ to $\infty$. Let $S^* \subseteq V$ be a \R2{non-empty} set of vertices such that $S^*\cup\{s\}$ is a minimum $s,t$-cut.
If $c(S^*\cup\{s\}) < n-k+1$, then 
$\sum_{e \in E(S^*)}  \sum_{i \in [k]} \bar x_{e,i} > |S^*|-1$.
Otherwise, all inequalities~\eqref{ineq:subtour} are satisfied.

\subsection{Partitioning formulations}\label{sec:form:partition}

Let $(G,w,k)$ be an instance of \textsc{min-max-bsf}.
Let $\mathcal{T}$ denote the set of all trees in $G$.
For each tree $T \in \mathcal{T}$, $w(T):=\sum_{e \in E(T)} w(e)$ denotes the weight of~$T$.
In the next formulation, we model a spanning $k$-forest of $G$ using binary variables~$x_T$ for all~$T \in \mathcal{T}$ that \R2{take} value 1 if $T$ is part of the forest.
\begin{align} 
        \min \: &  \omega &  \\
        \text{s.t.} \:& \sum_{T \in \mathcal{T}} x_T \le k,  & \label{ineq:k}\\
        & \sum_{T \in \mathcal{T} : v \in V(T)} x_T = 1 & \forall v \in V, \label{ineq:partition}\\
        &  \sum_{T \in \mathcal{T}:v\in V(T)} w(T)\:  x_{T}\leq \omega & \forall v \in V, \label{ineq:weight} \\
                & x_T \in \{0,1\} & \forall T \in \mathcal{T}.
\end{align}
Constraints~\eqref{ineq:k} ensure that at most $k$ trees are selected {(using inequality instead of equality does not change the optimal value)}, and constraints~\eqref{ineq:partition} guarantee that each vertex $v \in V$ is covered by exactly one tree. For $v\in V$, the left hand side of constraint~\eqref{ineq:weight} is the weight of the tree containing $v$. 
Clearly, the objective function together with these constraints 
correctly computes the weight of {a} heaviest tree. 
The pricing problem associated with this model is investigated in the following section.

A partitioning formulation for \textsc{max-min-bsf} can be obtained by changing the objective to maximization and replacing constraints \R2{\eqref{ineq:k}} and \eqref{ineq:weight} with 
\begin{align*}
    & \R2{\sum_{T \in \mathcal{T}} x_T\geq k,} &\text{and}\\
    &  \sum_{T \in \mathcal{T}:v\in V(T)} w(T)\:  x_{T}\geq \omega & \forall v \in V.   
\end{align*}

\section{A branch-and-price approach}\label{sec:bp}
Henceforth, we consider only the min-max version of \textsc{bsf} as
the max-min variant can be addressed in a  similar manner.
Since the number of variables in the \R2{set-partitioning} formulation is potentially exponential,  {we propose  a branch-and-price algorithm (see Algorithm \ref{alg:bp_rf}) to solve it.}
\R2{In the following algorithms, we assume, without loss of generality, that all weights are integers.}

{We  solve the LP relaxations at the nodes of the branch-and-price tree using a \R2{column-generation} algorithm (Algorithm \ref{alg:cg}). \R2{
We note here that in solving the pricing problems, we make use of the current upper bound and we do not generate columns associated with trees whose weights exceed this bound. As a result, we do not solve the LP relaxation of the \R2{set-partitioning} formulation presented above. This restriction preserves the integer optimum, since no tree whose weight exceeds a valid upper bound can belong to an optimal solution.}

Let $\mathcal{T}' \subseteq \mathcal{T}$ and let the corresponding \emph{restricted master problem} ($\mathrm{RMP}$) be the LP relaxation of the \R2{set-partitioning} formulation after setting $x_{T}=0$ for all $T\in \mathcal{T} \setminus \mathcal{T}'$. \R2{We maintain a global column pool shared across the entire branch-and-price tree.} We start the \R2{column-generation} algorithm with a feasible $\mathrm{RMP}$. At each iteration of the algorithm, we solve the current $\mathrm{RMP}$ to optimality and check if there are columns to add (columns with a negative reduced cost) by solving the pricing problem. At termination, which happens when pricing yields no columns to add, we have an optimal solution to the LP relaxation of the problem at the node. We check if the node can be pruned and if not, we branch. 
In the sequel, we describe the construction of the initial set of columns, 
the solution method for the pricing problem, the \R2{branching decisions} and how they are enforced, and the frequency with which we solve the integer version of the $\mathrm{RMP}$ to obtain feasible solutions.}


\subsection{Generation of the initial trees}
\label{subsec:GenerateTrees}

We construct the initial pool of columns for the branch-and-price algorithm in two stages. The first stage reuses the heuristic procedure described in Algorithm~\ref{alg:heur} in Section~\ref{sec:approximation}. 
A key difference with respect to the version in Section~\ref{sec:approximation} is that we keep track of the trees that appear in the sequence of forests produced during the search. 
Whenever a forest $F$ is generated, each  tree $T$ in $F$ with $w(T) \leq \mathrm{UB}$ is stored (if not already encountered) and offered to the master problem as a new column. This bookkeeping adds negligible overhead but significantly reduces the number of pricing iterations required later on.
The heuristic runs for a  duration that depends on the instance size.
{Precisely, we set the time limit to
\(
\min\{3^{\frac{n-20}{10}},27\}
\)
seconds, where $n$ is the number of vertices.} This exponential scaling allows larger instances to benefit from a longer construction phase without incurring excessive computational cost on smaller ones.

In the second stage, we enhance the initial pool through a randomized local modification strategy. At each iteration, we randomly select a previously generated tree $T$ and apply one of the following perturbations:
\begin{itemize}
    \item [\textbf{Vertex removal:}] Randomly remove a vertex $v \in V(T)$ and recompute a minimum spanning tree over the resulting subgraph, if it remains connected.
    \item [\textbf{Vertex addition:}] Randomly select a vertex $v \notin V(T)$, add it to $V(T)$, and recompute a minimum spanning tree on the induced subgraph, if connected.
\end{itemize}
If the resulting tree $T'$ satisfies $w(T') \leq \mathrm{UB}$ and has not been encountered before, it is stored and added as a new column. This randomized generation continues until a certain number of columns, defined by a threshold function $\left\lceil 2^{\min\{0.1n + 12, \:17\} - \frac{k}{2}} \right\rceil$, is reached. 
We also impose a stagnation limit: the process is terminated early if no new tree is added after 1000 consecutive iterations.
\R2{In addition to the trees generated by these two stages,
we include every singleton tree $(\{v\},\emptyset)$, for
$v \in V$, in the initial column pool $\mathcal{T}'$.}

\subsection{Pricing problem}\label{sec:pricingproblem}

The \R2{column-generation} algorithm iterates between solving the $\mathrm{RMP}$ and the pricing problem.   
{Let \( \theta \ge 0\), \( \eta_v  \in \Q\), and \( \zeta_v \ge 0\) for all $v \in V$ be the dual variables associated with constraints \eqref{ineq:k}, \eqref{ineq:partition}, and \eqref{ineq:weight}, respectively.  
Let~$(\hat{\theta},\hat{\eta}, \hat{\zeta})$ be an optimal solution to the dual of $\mathrm{RMP}$. 
The pricing problem consists in computing a tree  $T \in \mathcal{T} $ that maximizes  
\[
\rho(T) := -\hat{\theta} + \sum_{v \in V(T)} \hat{\eta}_v - w(T)\sum_{v\in V(T)}\hat{\zeta}_{v}.
\]
\R2{Note that $\rho(T)$ depends on the product of a weighted sum of the vertices of $T$, and the sum of weights of the edges in $T$.}
If there exists a tree $T \in \mathcal{T} $ with $\rho(T)>0$, then its corresponding variable $x_T$ has a negative reduced cost and introducing its column to the $\mathrm{RMP}$ may improve the optimal value. If no such tree exists, then the \R2{column-generation} algorithm terminates, \R2{and} the LP relaxation is solved to optimality. In this section, we explain how we solve the pricing problem. 
}

Notice that our pricing problem closely resembles 
the \emph{Prize-Collecting Steiner Tree}  (\pcst) problem, in which, given an undirected graph $G$, node prizes $p \colon V(G) \to \R2{\Q_{\geq 0}}$ and edge costs $c \colon E(G) \to \R2{\Q_{\geq 0}}$, one seeks a tree~$T$ in~$G$ that maximizes $\sum_{v\in V(T)}p(v) - \sum_{e\in E(T)}c(e)$.


However, there are two key differences. First, the node prizes \( \hat{\eta}_v \) in our pricing problem can be negative, whereas \pcst\ typically assumes \R2{non-negative} prizes (one of the few exceptions is the work of 
\cite{johnson2000prize}, which explicitly handles instances with negative prizes). Second, and more importantly, the weight of the tree is multiplied by a factor \( \alpha_T := \sum_{v \in V(T)} \hat{\zeta}_v \), which depends on the vertices included in the tree. 

The presence of negative prizes poses no major difficulty in our setting. However, the tree-dependent multiplicative coefficient in the objective function presents a significant 
challenge. In the following two subsections, we address this issue by reducing the problem to a collection of \emph{Budgeted Prize-Collecting Steiner Tree} instances.

\begin{problem}[\textsc{Budgeted Prize-Collecting Steiner Tree} (\bpcst)]\hfill\\
\textbf{Instance:} A graph $G$, a budget $B \in \mathbb{Z}_{\geq 0}$, vertex prizes $p \colon V(G) \to \Q$, edge costs $c\colon E(G) \to \R2{\Q_{\geq 0}}$, and edge weights $w\colon E(G) \to \R2{\mathbb{Q}_{\geq 0}}$.\\
\textbf{Objective:} Find a tree $T$ with  $w(T):=\sum_{e \in E(T)} w(e) \leq B$ that maximizes  
\[
\sum_{v \in V(T)} p(v) - \sum_{e \in E(T)} c(e).
\]  
\end{problem}





\subsubsection{Fixing the vertices}
\label{sec:FixSelectVertices}
The first approach to deal with the \R2{non-linearity} in the objective function of our pricing problem is to fix the set of vertices of the tree as follows. 
Let \( \mathcal{B} = \{v \in V : \hat{\zeta}_v > 0 \} \). For a given subset \( S \subseteq \mathcal{B} \), we enforce the inclusion of all vertices in \( S \) in the tree, while ensuring that all vertices in $\bar S :=\mathcal{B} \setminus S$ do not belong to this tree. 
This is done by setting the vertex prizes
\[
p(v) = \left\lbrace 
\begin{array}{lr}
    \hat\eta_v & \text{if } v \in V \setminus S, \\
    \hat\eta_v + M & \text{otherwise},
\end{array} 
\right.
\]
and the edge costs
\[
c(e) = \left\lbrace 
\begin{array}{lr}
    w_e \alpha_S & \text{if } e \cap \bar S = \emptyset, \\
    w_e \alpha_S + M & \text{otherwise},
\end{array} 
\right.
\]
where \( M \) is a sufficiently large number.

For each subset \( S \subseteq \mathcal{B} \), the transformed instance with vertex prizes \( p \) and edge costs \( c \) corresponds to \pcst. 
This implies that solving the pricing problem requires solving \( |\mathcal{S}| \) instances of \pcst, where \( \mathcal{S} := \{S \subseteq \mathcal{B}\} \). 
Interestingly, \( |\mathcal{B}| \) (and as a result $|\mathcal{S}|$) tends to be quite small for many of the iterations of the column generation, making this approach suitable in those iterations.
Since we are only interested in trees with weight not exceeding a given upper bound, we shall focus on solving these budgeted instances. 


Solving the pricing problem using this approach requires \( 2^{|\mathcal{B}|} \) iterations, as each subset \( S \subseteq \mathcal{B} \) defines an instance of the budgeted \pcst\ to be solved. In each iteration, an optimal tree \( T^*_S \) is obtained, and it is eligible to be added to the current pool \( \mathcal{T}' \) of columns in the \R2{column-generation} framework if and only if its reduced cost is negative, that is, \( \rho(T^*_S) > 0 \). Additionally, any feasible tree—such as incumbent solutions identified by the solver—with a negative reduced cost may also be added to \( \mathcal{T}' \).

If no such subset \( S \in \mathcal{S} \) yields a tree \( T^*_S \) with \( \rho(T^*_S) > 0 \), then no improving column remains, and the current  solution is optimal for the linear relaxation. 

\subsubsection{Fixing the weight}
\label{sec:FiXParam}

When \( |\mathcal{B}| \) is large, as the number of subsets \( |\mathcal{S}| \) grows  exponentially, the previous approach 
becomes computationally impractical. We therefore propose an alternative strategy by replacing the tree-dependent   term $w(T)$ with a fixed parameter $W$.

We define the approximate \R2{reduced-cost} function as
\[
    \rho'(W,T) = -\hat{\theta} + \sum_{v \in V(T)} \hat{\eta}_v - W \sum_{v \in V(T)} \hat{\zeta}_v,
\]
where \(W \in [0,\mathrm{UB}] \) is a fixed parameter. Note that $\rho(T) = \rho'(w(T),T)$. Maximizing \( \rho'(W,T) \) corresponds to solving an instance of the \emph{Maximum Weight Connected Subgraph} (\mwcs) problem. Given an undirected graph $G$ and node prizes $p \colon V(G) \to \Q$, \mwcs\ consists in finding a connected subgraph $T$ of $G$ that maximizes $\sum_{v\in V(T)}p(v)$.
In our setting, however, we impose an additional budget constraint on the tree weight: we seek a tree \( T \) maximizing \( \rho'(W,T) \) such that \( w(T) \leq W \). 
This problem is a particular case of the  \bpcst, as shown next.

One may reduce a budgeted \mwcs\ instance to a \bpcst\ instance 
by setting the following node prizes and edge costs: 
\(p^{}_W(v) = \hat{\eta}_v - W \hat{\zeta}_v - \gamma(W)\) for all   \(v \in V\), and \(c^{}_W(e) = - \gamma(W)\)  for all \(e\in E\),
where $\gamma(W) = \min\left(\min_{u \in V} \left\{\hat{\eta}_u - W \hat{\zeta}_u\right\}, 0\right)$. 
For any tree $T$, it clearly holds that 
\[
    \rho'(W,T) = -\hat{\theta} + \sum_{u \in V(T)} p^{}_W(u) - \sum_{e \in E(T)} c^{}_W(e)  + \gamma(W).
\]

Given \( W\in [0, \mathrm{UB}] \), let \( T^*_W \) denote  an optimal solution to the corresponding \bpcst\ instance. 
If \( w(T^*_W) = W \), then 
$\rho'(W,T^*_W) = \rho'(w(T^*_W),T^*_W)  = \rho(T^*_W).$
Otherwise, it holds that 
\(
\rho'(W,T^*_W) \leq \rho'(w(T^*_W),T^*_W) = \rho(T^*_W).
\)
The tree \( T^*_W \) can be included in \( \mathcal{T}' \) if \( \rho(T^*_W) > 0 \), as this implies a negative reduced cost. As a particular example, when \( \rho'(W,T^*_W) > 0 \), it follows that \( \rho(T^*_W) > 0 \). Note, however, that even if \( \rho'(W,T^*_W) \leq 0 \), it is still possible that \( \rho(T^*_W) > 0 \).



To search for improving columns, we apply an iterative procedure over the parameter \( W\), starting from the current upper bound \( \mathrm{UB} \). At each iteration, we solve the transformed \bpcst\ instance with the current value of \( W \), obtaining an optimal tree \( T^*_W\). If \( \rho(T^*_W) > 0 \), then \( T^*_W \) is added to the pool of candidate trees, and the parameter is updated as \( W \gets w(T^*_W) - 1 \), thereby reducing the budget and guiding the search toward smaller subtrees that may yield better reduced costs. If \( \rho(T^*_W) \le 0 \), we instead set \( W \gets W - 1 \) and proceed to the next iteration.

\R2{This iterative process continues while $W\geq W_{\min}$, where $W_{\min}=0$ if $G$ contains a zero-weight edge and $W_{\min}=1$ otherwise.} If no tree \( T^*_W \) with \( \rho'(W, T^*_W) > 0 \) is found for any \( \R2{W \in [W_{\min}, \mathrm{UB}]} \), we conclude that no improving column exists under this subprocedure, and the current solution to the master problem is optimal for the linear relaxation. 


\subsubsection{Solving the budgeted PCST}
Regardless of whether we fix the set of vertices in the tree or the weight of the tree, we ultimately solve instances of the Budgeted Prize-Collecting Steiner Tree problem. 


Although numerous reduction techniques exist in the literature for the classical \pcst~\citep{ljubic2006algorithmic,rehfeldt2022exact}, the additional budget constraint renders most of these techniques inapplicable to \bpcst.
Moreover, the few reduction techniques that we could adapt to the budgeted setting are not very effective in practice because the time required to compute the reduction frequently exceeds the time saved during the subsequent optimization.
Hence these reduction rules are not applied in the proposed solving method.

We next present a natural integer {linear programming} formulation for \bpcst\ using arborescences, that is, rooted trees with all \R2{arcs} oriented away from the root.
Let $(G,B,p,c,w)$ be an instance of \bpcst.
First we transform the undirected graph \( G = (V, E) \) into a directed graph.
This transformation begins by adding a new vertex \( r \), which serves as the root of the arborescence. 
Each undirected edge \( \{u,v\} \in E \) is replaced with two directed arcs, \( (u,v) \) and \( (v,u) \), with the arc profit and weight defined as:
\(
c'(uv) = p(v) - c(\{u,v\}), \: c'(vu) =p(u) - c(\{u,v\}), \text{ and } w'(uv) = w'(vu) = w(\{u,v\}).
\)
Furthermore, we introduce additional arcs \( (r,v) \) for each \( v \in V \), with the corresponding profits and weight values: $c'(r,v) = p(v) - \theta$ and  $w'(rv) = 0$ for all $v \in V$.
As a result, we obtain a directed graph \( G' = (V\cup\{r\}, A) \) with arc profits and weights~$c'$ and $w'$, respectively. 
For each vertex $v \in V\cup\{r\}$, \( \delta^+(v) \) and \( \delta^-(v) \) denote the sets of outgoing and incoming arcs of \( v \)  in $G'$. 
\begin{align}
    \max \:\:& \sum_{a \in A} c'(a) \, x_a \\ 
    \text{s.t.} \:\: & \sum_{a \in A} w'(a) \, x_a \le B, &\\ 
    & \sum_{a \in \delta^+(r)} x_a = 1, \\ 
    & \sum_{a \in \delta^-(v)} x_{a} = y_v & \forall v \in V, \\ 
    & \sum_{a \in \delta^-(S)} x_a \ge y_v & \forall S \subseteq V, \, v \in S, \label{ABOR:Sets} \\ 
    & x_a\in \{0,1\} & \forall a \in A, \\
    & y_v \in \{0,1\}& \forall v \in V.
\end{align}

Note that there are exponentially many inequalities~\eqref{ABOR:Sets}.
These inequalities are separated in polynomial time by computing minimum cuts between the root and each vertex $v \in V$ with $y_v > 0$, similarly to the lazy \R2{constraint-generation} approach by \citet{ljubic2006algorithmic}. However, unlike \citeauthor{ljubic2006algorithmic}, we do not perform nested constraint generation, as the sizes of our \bpcst\ instances are relatively small.

While solving the \bpcst, we retain all incumbent solutions in the hope of identifying multiple trees \( T \) with \( \rho(T) > 0 \). To guide the search, we also impose 
\begin{align}
    \sum_{a \in A} c'(a) \, x_a \ge 0. \label{F:guide}
\end{align}
{This inequality is not intended to strengthen the formulation in the usual polyhedral sense; rather, if the optimal value of \R2{this} model is smaller than $\varepsilon$ or negative, then \eqref{F:guide} renders the model infeasible. This infeasibility certifies that no improving tree exists. 
}

\subsubsection{Choice of the method}
{We presented two different methods to solve the pricing problem, namely, fixing the vertices and fixing the weight. Let~$(\hat{\theta},\hat{\eta}, \hat{\zeta})$ be an optimal solution to the dual of $\mathrm{RMP}$. }
Depending on the cardinality of the set \( \mathcal{B} = \{v \in V : \R2{\hat{\zeta}_v} > 0\} \), we proceed as follows:
\begin{enumerate}[(i)]
    \item  we select sets of vertices to be in the tree as in Section~\ref{sec:FixSelectVertices} (procedure~\texttt{PCST\_FV}) if \( 2^{|\mathcal{B}|} < 2 \cdot \mathrm{UB} \), 
    \item we choose weights for the trees as in Section~\ref{sec:FiXParam} (procedure~\texttt{PCST\_FW})  otherwise.
\end{enumerate}
{The first approach is generally preferred because we can use any of the fixed vertices as the root in the formulation for \bpcst\ (by imposing that the solution uses the arc from $r$ to this node), which reduces symmetry as a tree 
admits equivalent alternative representations depending on the selected root,
and improves the performance of the branch-and-cut algorithm.}

\subsection{Branching and after}
\label{sec:BS}

{
If the optimal solution of an LP relaxation at a node of the branch-and-price tree is fractional and the node cannot be pruned, then we branch using the \emph{\R2{Ryan--Foster}} \citep{RyanFoster1981} rule, which is particularly \R2{well suited} for \R2{set-partitioning and -covering} problems \citep{mehrotra1996column, wahlen2025branch}. For a pair of vertices whose corresponding variables  have fractional values, we create two child nodes by enforcing the two vertices to belong to the same tree in one (the \R2{\textsc{same}} node) and to different trees in the other one (the \R2{\textsc{differ}} node).} 


{After branching at node \(\mu\), we immediately evaluate
the relaxations of both children nodes and add the non-pruned nodes to the queue \(\mathcal{Q}\). This ensures that every unvisited node is associated with an evaluated lower bound, allowing us to implement a best-bound node-selection strategy. 
}





\begin{algorithm}[t]
\caption{Branch-and-Price with Ryan--Foster Branching}
\label{alg:bp_rf}
\begin{algorithmic}[1]
\State Initialize \( \mathrm{UB} \) as the upper bound obtained from tree generation
\State Initialize the node queue \( \mathcal{Q} \gets \emptyset \)
\State Create the root node \( \mu = 0 \) with \R2{branching decisions} \( R_\mu \gets \emptyset \)

\State \( (\mathrm{LB}_{\mu}, \mathrm{UB}, b_\mu, s_\mu)
    \gets \) Algorithm~\ref{alg:cg} on node \( \mu \)
    \Comment{Solve root node by column generation}

\If{\( s_\mu = \textsc{TimeLimitReached} \)}
     \Return \( (\mathrm{UB}, \R2{\rm{NULL}}, \textsc{TimeLimitReached}) \)
    \Comment{gap is not defined}
\EndIf

\If{\( b_\mu \neq \emptyset \) and \( \R2{\mathrm{LB}_\mu < \mathrm{UB}} \)}
     Insert node \( \mu \) into \( \mathcal{Q} \)
\EndIf

\While{\( \mathcal{Q} \neq \emptyset \)}

    \State Pop node \( \mu \) from \( \mathcal{Q} \) with smallest \( \mathrm{LB}_\mu \)

    \If{time limit is reached}   
 \Return \( \left(\mathrm{UB}, \dfrac{\mathrm{UB}- 
 \mathrm{LB}_{\mu}
 }{\mathrm{UB}}, \textsc{TimeLimitReached}\right) \)
    \EndIf

    \If{\( \R2{\mathrm{LB}_\mu \ge \mathrm{UB}} \)}
      \textbf{continue} 
    \EndIf
    
    \State Let \( (u,v) \gets b_\mu \)

    \ForAll{\( r \in \{\R2{\textsc{same}}, \R2{\textsc{differ}}\} \)}
        \State Create child node \( \mu' \)
        \State \( R_{\mu'} \gets R_\mu \cup \{(u,v,r)\} \)

        \State \( (\mathrm{LB}_{\mu'}, \mathrm{UB}, b_{\mu'}, s_{\mu'})
            \gets \) Algorithm~\ref{alg:cg} on node \( \mu' \)

        \If{\( s_{\mu'} = \textsc{TimeLimitReached} \)}
           \Return \( \left(\mathrm{UB}, \dfrac{\mathrm{UB}-
           \mathrm{LB}_{\mu}
           }{\mathrm{UB}}, \textsc{TimeLimitReached}\right) \)
        \EndIf

        \If{\( b_{\mu'} \neq \emptyset \) and \( \R2{ \mathrm{LB}_{\mu'} < \mathrm{UB}} \)} 
            Insert node \( \mu' \) into \( \mathcal{Q} \)
        \EndIf
    \EndFor
\EndWhile

\State \Return \( (\mathrm{UB}, 0, \textsc{Optimal})\) 
\end{algorithmic}
\end{algorithm}


{
For each node $\mu$ of the branch-and-price tree, let $R_\mu$ denote the set of \R2{branching decisions} associated with~$\mu$. This set is empty at the root node, i.e.,  \( R_0 = \emptyset \). }
{
The \R2{branching decisions} are enforced in two ways. First, in the inherited RMP, variables corresponding to previously generated columns that violate the \R2{branching decisions} at node $\mu$ are locally fixed to zero (see Algorithm~\ref{alg:cg}, line 1). Second, in the pricing problem, for every pair $(u,v)$ branched \R2{\textsc{same}}, we add the constraint
\(
y_u-y_v=0,
\)
and, for every pair $(u,v)$ branched \R2{\textsc{differ}}, we add the constraint
\(
y_u+y_v\leq 1.
\)
Thus, every newly generated column is compatible with the \R2{branching decisions} at node~$\mu$.
}

{After enforcing the \R2{branching decisions} and filtering out the columns that do not satisfy them, the  $\mathrm{RMP}$ may become infeasible.  
In this case, 
we add artificial columns.
More precisely, we construct a collection of $k$ artificial columns whose supports form a partition of $V$ and satisfy the branching decisions at node $\mu$.
An artificial column that does not  correspond to a tree is assigned a sufficiently large penalty in the objective function. These columns are used solely to initialize a feasible RMP and to obtain dual multipliers for the pricing problem.}


\begin{algorithm}[t!]
\caption{Column Generation at node \( \mu \) for \( \mathcal{T}' \)}
\label{alg:cg}
\begin{algorithmic}[1]
\State {Update the RMP by locally fixing to zero the variables associated with
columns in $\mathcal{T}'$ that violate the \R2{branching decisions} at node $\mu$.}
\While{\textbf{true}}
    \If{time limit is reached}
         \Return \( (\emptyset,\mathrm{UB},\emptyset,\textsc{TimeLimitReached}) \)
        \Comment{No LB nor branch pair}
    \EndIf
    \State Solve $\mathrm{RMP}(\mathcal{T}')$ to get \( (\omega^*, x^*, \R2{\hat{\theta}}, \R2{\hat{\eta}}, \R2{\hat{\zeta}}) \)
     \If{$\mathrm{RMP}(\mathcal{T}')$ is infeasible}
        \State Add feasibility columns compatible with the \R2{branching decisions} at node \( \mu \)
        \State \textbf{continue} 
    \EndIf
    \State Let \( \R2{\mathcal{B}} \gets \{v \in V : \R2{\hat{\zeta}_v} > 0\} \)
    \If{\( 2^{|\R2{\mathcal{B}}|} < 2 \cdot \mathrm{UB} \)}
         \( \mathcal{T}_\mu \gets \texttt{PCST\_FV}\) on node $\mu$
    \Else \:
         \( \mathcal{T}_\mu \gets \texttt{PCST\_FW} \) on node $\mu$
    \EndIf
    \If{{\( \mathcal{T}_\mu \neq \emptyset \)} }
        \State \( \mathcal{T}' \gets \mathcal{T}' \cup \mathcal{T}_\mu \)
        \State Add new trees to the $\mathrm{RMP}(\mathcal{T}')$
        \If{\( \mu \) is the root node and UB update criteria is met}
             update \(\mathrm{UB}\)
        \EndIf
    \Else \: \textbf{break}
    \EndIf
\EndWhile
\State \( \mathrm{LB}_\mu \gets \omega^* \)
\State \(  b_\mu  \gets \emptyset \)
\If{\( \R2{\mathrm{LB}_\mu < \mathrm{UB}} \)}
    \State Let \( u, v \in V \) with $u\neq v$ such that  $|q_{uv} -0.5|$ is minimized \Comment{Most fractional pair}
    \If{$q_{uv}$ is not integer}
         Set branching pair \( b_\mu \gets (u, v) \)
    \Else
        {\; Compute an integer solution using the construction in the proof of Proposition~\ref{prop:ryanfoster}, and
        update~\( \mathrm{UB} \)}
    \EndIf
\EndIf
\State \Return \( (\mathrm{LB}_\mu,\mathrm{UB},b_\mu,\textsc{Complete}) \)
\end{algorithmic}
\end{algorithm}




{
We next prove that branching on individual variables is not necessary when using the \R2{Ryan--Foster} branching rule.

\begin{proposition}
\label{prop:ryanfoster}
Let \(x^*\) be an optimal solution of the RMP associated with any node of the branch-and-bound tree. 
If no fractional \R2{Ryan--Foster} pair exists, that is,
\[
q_{uv}:=\sum_{T\in\mathcal{T}' : u,v\in V(T)} x_T^* \in \{0,1\}
\qquad \forall u,v\in V \text{ with } u\neq v,
\]
\R2{then an integer optimal solution of the RMP can be
constructed from $x^*$ without branching on individual
master variables.}
\end{proposition}

\proof{Proof.}
\R2{If $x^*$ is already integral, then we are done.
Henceforth, assume that $x^*$ is fractional.
We next show how to construct an integer solution from $x^*$ with the same objective function value.}
Let $\mathcal{T}'_u=\{T\in\mathcal{T}' : u\in V(T)\}$, and let $\mathcal{T}'_{uv}=\{T\in\mathcal{T}' : u,v\in V(T)\}$ for all $u,v\in V$ with $u\neq v$. 
Suppose that $q_{uv}\in \{0,1\}$ for all  $u,v\in V$ with $u\neq v$.
Let $\mathcal{T}^f=\{T\in \mathcal{T}': 0<x^*_{T}<1\}$. 
Due to constraints \eqref{ineq:partition}, there exist at least two trees in $ \mathcal{T}^f$ with fractional values in $x^*$, and so there is~$T' \in \mathcal{T}^f$ containing \R2{at least} two vertices.   
Let $u$ and $v$ be distinct vertices in $V(T')$, and note that $q_{uv}=1$. 
Since $\mathcal{T}'_{uv}\subseteq \mathcal{T}'_u$, it holds that $\sum_{T\in\mathcal{T}'_{u}} x_T^* =1$, and so $x^*_T=0$ for all $T \in \mathcal{T}'_u\setminus \mathcal{T}'_{uv}$. 
Repeating this argument for every tree in $\mathcal{T}^f$ and all pairs of vertices in this tree,  we conclude that \R2{either} $V(T_1)=V(T_2)$ or $V(T_1)\cap V(T_2)=\emptyset$ for all trees $T_1, T_2\in \mathcal{T}^f$. 
\R2{Let  \(\{V_1,\ldots,V_t\}\), for a positive integer $t$,  be the collection of all distinct vertex supports of trees in \(\mathcal T^f\), and
note that $\sum_{T \in \mathcal{T}^f : V(T)=V_i} x^*_T = 1$ for all $i \in [t]$.}
\R2{Finally, one may obtain an integer solution with the same objective value of $x^*$ by taking any tree in $\arg\min_{T \in \mathcal{T}_i^f} w(T)$ where $\mathcal{T}^f_i:=\{T \in \mathcal{T}^f : V(T)=V_i\}$ for each $i \in [t]$.}
%

}

\subsection{Primal bounds}
{At the end of the column-generation procedure at any node of the branch-and-bound tree, we solve the integer version of the RMP over all columns generated at that node and update the incumbent and the upper bound whenever a better feasible solution is found. Moreover, at the root node, we solve the integer version of the RMP whenever the following conditions hold: (1) the optimal value of the RMP satisfies $\omega^* \leq \mathrm{UB}-1$; and (2) either the number of columns added since the previous update exceeds a predefined threshold, a prescribed number of column-generation iterations has elapsed, or a given time limit has been reached.}

\section{Computational experiments for \textsc{min-max-bsf}} \label{sec:experiments}
All computational experiments were conducted on a machine running Microsoft Windows 11 Education, equipped with an Intel\textsuperscript{\textregistered} Core\textsuperscript{TM} i9-9900 CPU @ 3.10\,GHz, featuring 8 physical cores and 16 logical processors, and 64\,GB of installed RAM. The C++ implementation of our solution methods was compiled using the Microsoft Visual C++ compiler, and uses the Boost Library \texttt{1.80.0}, LEMON Library \texttt{1.3.1}, and IBM ILOG CPLEX Optimization Studio \texttt{22.1.1}.

All CPLEX-based methods were configured to run with a single thread and a relative optimality gap tolerance of~\(10^{-5}\). All other CPLEX parameters were set to default. To ensure numerical stability, we applied a tolerance of~\(10^{-3}\) when interpreting solution values of continuous variables.


To evaluate the computational performance of the proposed formulations and algorithms, we generated synthetic instances of \textsc{min-max-bsf} using the Boost Graph Library.

The dataset contains 380 instances generated as follows.
For each \( n \in \{20, 30, 40, 50\} \), \( p \in \{0.1, 0.2, 0.3, 0.4, 0.5\} \), and $k \in \{2,4,6,8,10\}$,
we create four instances, each of them consisting of a random graph with $n$ vertices and $m=\left\lfloor p \cdot \frac{n(n-1)}{2} \right\rfloor$ edges, generated using \texttt{boost::generate\_random\_graph()}.
Any disconnected graph is discarded. 
Thus all graphs in the dataset are connected.
We note that the combination $n=20$ and $p=0.1$ always led to a disconnected graph, this explains the total number of instances.
Finally, the weight of each edge in these instances is sampled uniformly at random in the set $\{1,\ldots,100\}$.

In what follows, we denote by $\mathcal{F}$ the branch-and-bound algorithm based on the compact flow formulation presented in Section~\ref{sec:form:flow}, $\mathcal{E}$ the branch-and-cut algorithm using the \R2{cycle-elimination} model of Section~\ref{sec:form:elimination}, and $\mathcal{P}$ the branch-and-price algorithm based on the model in Section~\ref{sec:form:partition} and described in Section~\ref{sec:bp}.
We use the solution generated by Algorithm~\ref{alg:heur} with a time limit of {one~second} as a warm-start solution to $\mathcal{F}$  and~$\mathcal{E}$.
{We impose a time limit of \(1{,}800\) seconds for small instances with \(n \leq 20\), and \(10{,}800\) seconds for larger instances with \(n \geq 30\), uniformly across all solution methods.}

The test instances, computational results, and source codes of the complete implementation are 
available at
\texttt{https://github.com/INFORMSJoC/2025.1551}.

\subsection{Comparison of the linear relaxations} 

We first evaluate the quality of the dual bounds given by the linear relaxation of the three proposed formulations in Section~\ref{sec:formulations} with respect to the \R2{best-known} primal bounds obtained by any of the solving methods $\mathcal{F}$, $\mathcal{E}$,  and $\mathcal{P}$.
For each of the three formulations, we compute the average relative optimality gaps of these bounds across all instances in each class. 

Table~\ref{tab:bound-performance} summarizes these results, where columns $F_{\rm{gap}}$, $E_{\rm{gap}}$, and $P_{\rm{gap}}$ indicate the average gap percentages obtained with the linear relaxation of the flow model (Section~\ref{sec:form:flow}), \R2{cycle-elimination} model (Section~\ref{sec:form:elimination}), and partition model (Section~\ref{sec:form:partition}), respectively. \R2{For \(P_{\rm gap}\), the reported bounds are obtained with pricing restricted by the incumbent upper bound, as described in Section~\(\ref{sec:bp}\), and may therefore differ from those of the unrestricted partitioning formulation.}
It is clear from Table~\ref{tab:bound-performance} that the \R2{cycle-elimination} formulation has the worst average dual bound in all classes of instances, and that the partition formulation produces the best average dual bounds.
{
A caveat is that the column-generation algorithm for the partition formulation does not finish solving the LP relaxation within the time limit for 12 out of 380 instances, primarily because the associated pricing problems are particularly difficult to solve. These instances are not considered in the averages for this model. 
}

\begin{table}[t!]
\centering
\caption{Average relative optimality 
gaps (\%) for different bounds across instance classes. The column generation does not finish within the time limit for the number of instances indicated in parentheses, and so these instances are excluded from the average.}
\label{tab:bound-performance}
\begin{tabular*}{0.5\linewidth}{@{\extracolsep{\fill}} c c   r r r  }
\toprule
$n$ & \#\text{ instances} &\( F_{\rm{gap}} \) & \( E_{\rm{gap}} \) & \( P_{\rm{gap}} \)  \\
\midrule
20 & 80 & 38.92 & 47.30 & 1.74               \\
30 & 100 & 26.06 & 35.62 & 2.67               \\
40 & 100 & 20.65 & 32.12 &(3)\tnote{*} 3.52    \\
50 & 100 & 15.94 & 27.89     &(9)\tnote{*} 2.46      \\
\bottomrule
\end{tabular*}
\end{table}

\subsection{Comparison of the solving methods}

The weak dual bounds yielded by the \R2{cycle-elimination} formulation (see Table~\ref{tab:bound-performance}) seem to have a significant impact on the performance of the solving method~$\mathcal{E}$.
Even on the smallest instances (i.e., $n=20$) in our dataset, $\mathcal{E}$  was not able to solve 3 (out of 80) of these instances, while both  $\mathcal{F}$ and $\mathcal{P}$ solved each of them within 10 seconds on average.
Considering the instances with $n=30$, $\mathcal{E}$ does not solve 74 (out of 100) of these instances, while both  $\mathcal{F}$ and $\mathcal{P}$ solved all of them.
Due to this very poor performance of $\mathcal{E}$ compared to $\mathcal{F}$ and $\mathcal{P}$, we do not report its computational results on the entire dataset.

In the remainder of this section, we focus on the comparison of $\mathcal{F}$ and $\mathcal{P}$ considering all 380 instances in the dataset.
We first use performance profiles introduced by \cite{dolan2002benchmarking} to compare the running times of these two algorithms. 
Let $A$ be a set of algorithms and $P$ a \R2{non-empty} set of instances.
For each instance $p \in P$ and algorithm $a \in A$, $t_{p,a}$ denotes the \R2{total running time} required by $a$ to solve $p$.
The performance ratio of an algorithm~$a \in A$ on an instance~$p \in P$ is defined as~$r_{p,a} = t_{p,a}/\min\{t_{p,b} :  b \in A\}$.
In case $p$ is not solved by $a$, we set $r_{p,a}=+\infty$.
For every $a \in A$ and $\tau \geq 1$, define
\(\rho_{a}(\tau) = |\{p \in P : r_{p,a}\leq \tau\}|/|P|. \)
The performance profile of an algorithm $a\in A$ corresponds to the cumulative distribution function~$\rho_a$.  

The performance profiles of $\mathcal{F}$ and $\mathcal{P}$ on the entire dataset \R2{are} depicted in Figure~\ref{fig:perf-profile}.
Note that $\mathcal{F}$ and $\mathcal{P}$ are the fastest solving methods for approximately \R2{$49\%$} and \R2{$44\%$} of the instances, respectively. 
\R2{At $\tau=100$, the profiles reach approximately $74\%$ for $\mathcal{F}$ and $84\%$ for $\mathcal{P}$.}


\perfset{data-runtime.tex}




    \begin{figure}[t!]
        \centering
        \begin{subfigure}[t]{0.45\textwidth}
        \centering
        \begin{tikzpicture}
            \begin{axis}[
            height=7cm,
            grid=both, 
            no marks, xlabel={Normalized time $\tau$}, ylabel={Proportion of  instances $\rho$}, xmin=1, xmax=20,
            extra x ticks={1},
            label style={font=\small},
            tick label style={font=\small},    
            title style={font=\small},
            legend style={font=\footnotesize, at={(0.95,0.3)},anchor=east}]
                \addprofiles{2}{100}
                \legend{$\mathcal{P}$, $\mathcal{F}$}
            \end{axis}
        \end{tikzpicture}
        \end{subfigure}
        \hfill
        \begin{subfigure}[t]{0.45\textwidth}
        \centering
        \begin{tikzpicture}            
            \begin{axis}[
            height=7cm,
            grid=both,no marks,
            xlabel={Normalized time $\tau$}, 
            xmin=1, 
            xmax=100,
            extra x ticks={1},
            label style={font=\small},
            tick label style={font=\small},    
            title style={font=\small},
            legend style={font=\footnotesize, at={(0.95,0.3)},anchor=east}]
                \addprofiles{2}{100}
                \legend{$\mathcal{P}$, $\mathcal{F}$}
            \end{axis}
        \end{tikzpicture}
        \end{subfigure}%
        
        \caption{\R2{Performance profiles of the algorithms $\mathcal{F}$ and $\mathcal{P}$ on the entire dataset. {The left plot is a zoomed-in view of the right plot.}}} 
        \label{fig:perf-profile}
    \end{figure}

We observe that $\mathcal{F}$ performs very well on instances with $k=2$ and on the smallest instances, i.e.,  $n=20$.
To better understand how these parameters affect the performances of the proposed methods, 
we consider two subsets of the dataset: (i) instances with $k\ge 4$, and (ii) instances with $n\ge 30$

Figure~\ref{fig:perf-profile-random-large-k} shows the performance profiles on the instances with $k\geq 4$.
\R2{On this subset, $\mathcal{F}$ and $\mathcal{P}$ solve approximately $68.1\%$ and $89.8\%$ of the instances, respectively.}
Additionally, $\mathcal{F}$ and $\mathcal{P}$ are the fastest algorithms on \R2{$37\%$} and \R2{$55\%$} of these problems, respectively.
The performances on the instances with $n\geq 30$ are illustrated in Figure~\ref{fig:perf-profile-random-large-n}.
Again, $\mathcal{P}$ clearly outperforms $\mathcal{F}$ in this case, solving over $80\%$ of these instances.
We next discuss these results in more \R2{detail}.


\perfset{data-runtime-large-k.tex}




    \begin{figure}[t!]
        \centering
        \begin{subfigure}[t]{0.45\textwidth}
        \centering
        \begin{tikzpicture}
            \begin{axis}[
            height=7cm,
            grid=both, 
            no marks, xlabel={Normalized time $\tau$}, ylabel={Proportion of  instances $\rho$}, xmin=1, xmax=20,
            extra x ticks={1},
            label style={font=\small},
            tick label style={font=\small},    
            title style={font=\small},
            legend style={font=\footnotesize, at={(0.95,0.3)},anchor=east}]
                \addprofiles{2}{100}
                \legend{$\mathcal{P}$, $\mathcal{F}$}
            \end{axis}
        \end{tikzpicture}
        \end{subfigure}
        \hfill
        \begin{subfigure}[t]{0.45\textwidth}
        \centering
        \begin{tikzpicture}            
            \begin{axis}[
            height=7cm,
            grid=both,no marks,
            xlabel={Normalized time $\tau$}, 
            xmin=1, 
            xmax=100,
            extra x ticks={1},
            label style={font=\small},
            tick label style={font=\small},    
            title style={font=\small},
            legend style={font=\footnotesize, at={(0.95,0.3)},anchor=east}]
                \addprofiles{2}{100}
                \legend{$\mathcal{P}$, $\mathcal{F}$}
            \end{axis}
        \end{tikzpicture}
        \end{subfigure}%
        
        \caption{\R2{Performance profiles of the algorithms $\mathcal{F}$ and $\mathcal{P}$ on the instances with $k\geq 4$.}}
        \label{fig:perf-profile-random-large-k}
    \end{figure}

\perfset{data-runtime-large-n.tex}




    \begin{figure}[t!]
        \centering
        \begin{subfigure}[t]{0.45\textwidth}
        \centering
        \begin{tikzpicture}
            \begin{axis}[
            height=7cm,
            grid=both, 
            no marks, xlabel={Normalized time $\tau$}, ylabel={Proportion of  instances $\rho$}, xmin=1, xmax=20,
            extra x ticks={1},
            label style={font=\small},
            tick label style={font=\small},    
            title style={font=\small},
            legend style={font=\footnotesize, at={(0.95,0.3)},anchor=east}]
                \addprofiles{2}{100}
                \legend{$\mathcal{P}$, $\mathcal{F}$}
            \end{axis}
        \end{tikzpicture}
        \end{subfigure}
        \hfill
        \begin{subfigure}[t]{0.45\textwidth}
        \centering
        \begin{tikzpicture}            
            \begin{axis}[
            height=7cm,
            grid=both,no marks,
            xlabel={Normalized time $\tau$}, 
            xmin=1, 
            xmax=100,
            extra x ticks={1},
            label style={font=\small},
            tick label style={font=\small},    
            title style={font=\small},
            legend style={font=\footnotesize, at={(0.95,0.3)},anchor=east}]
                \addprofiles{2}{100}
                \legend{$\mathcal{P}$, $\mathcal{F}$}
            \end{axis}
        \end{tikzpicture}
        \end{subfigure}%
        
        \caption{\R2{Performance profiles of the algorithms $\mathcal{F}$ and $\mathcal{P}$ on the instances with $n\geq 30$.}}
        \label{fig:perf-profile-random-large-n}
    \end{figure}

The algorithm based on the compact flow formulation performs very well for small values of \( k \), particularly when \( k =2 \), with low running times and solving all such instances. 
For example, with \( n = 30 \) and \( k = 2 \), the $\mathcal{F}$  solves all 20 instances in under 10 seconds on average. 
In such cases, $\mathcal{P}$  is significantly slower, with its performance hindered by the overhead from LP column generation and pricing, even though it still solves all these instances.

As the value of \( k \) increases, the computational advantage of the compact formulation diminishes rapidly. 
Already when \( k = 4 \), $\mathcal{F}$  begins to show signs of computational stress, with some instances reaching high running times or nearing the time limit. 
The situation becomes more pronounced with \( k \ge 6 \), where the capacity of $\mathcal{F}$  to solve instances degrades substantially. 
For \( n = 50 \) and \( k \geq 6 \), $\mathcal{F}$ solves very few instances to optimality (in some cases none), and the average runtimes \R2{are} close to the time limit.
In contrast, $\mathcal{P}$ exhibits more stable behavior as \( k \) increases. Although it may not outperform $\mathcal{F}$ in raw speed for small \( k \), it proves significantly more robust in closing the optimality gap for larger and more complex instances. 
Notably, for \( n = 50 \), $\mathcal{P}$ consistently solves a larger number of instances than $\mathcal{F}$ for all values of \( k > 2 \). 
For example, with \( k = 10 \), \R2{$\mathcal{P}$} solves 15 out of 20 instances, while $\mathcal{F}$ does not solve any.

One of the key reasons for the superior performance of the branch-and-price approach with large \( k \) is structural: as \( k \) increases, the optimal solutions tend to consist of smaller and lighter trees. 
This structure limits the number of heavy trees that could potentially be generated as columns in the pricing problem. 
The upper bound acts as a filter that effectively prevents the inclusion of many large or suboptimal trees, thereby reducing the combinatorial explosion in the master problem and improving the overall convergence. 
As a result, $\mathcal{P}$ scales more gracefully with increasing \( k \), both in terms of computational effort and the quality of the lower bounds it provides.

Moreover, $\mathcal{P}$’s ability to provide stronger dual bounds and 
smaller final integrality gaps is evident in Table~\ref{tab:bound-performance} and 
Figure~\ref{fig:boxplot}. The \emph{final integrality gap} is computed as the relative difference 
between the best primal and dual bounds at termination. 
While the gaps produced by $\mathcal{F}$ increase sharply with \( k \) and 
\( n \), \R2{the median of the 25 positive gaps plotted for $\mathcal{P}$ is $3\%$}. 
This is particularly important when evaluating solution quality of unsolved 
instances, as $\mathcal{P}$ provides tight lower bounds, which are useful 
in practical decision-making.

\begin{figure*}[t]
\centering
\begin{subfigure}[t]{0.45\textwidth}
\centering
\begin{tikzpicture}
	\pgfplotstableread[col sep=comma]{gap-bp.csv}\csvdataA
	\pgfplotstabletranspose\datatransposedA{\csvdataA}

	\pgfplotstableread[col sep=comma]{gap-c.csv}\csvdataB
	\pgfplotstabletranspose\datatransposedB{\csvdataB}


    
	\begin{axis}[
        y=0.25cm,
        x=1cm,
        xminorgrids=true,
		boxplot/draw direction = y,
		x axis line style = {opacity=0},
		axis x line* = bottom,
		axis y line* = left,
        enlarge y limits={abs=0.4cm},
		ymajorgrids,
		xtick = {1, 2, 3, 4},
		xticklabel style = {align=center, font=\footnotesize, rotate=0},
		xticklabels = {$\mathcal{P}$, $\mathcal{F}$},
		xtick style = {draw=none}, 
		ylabel = {Gap \%},
        boxplot/whisker range={100},
        label style={font=\small},
        tick label style={font=\small},    
        title style={font=\small},
        legend style={font=\footnotesize},
        x tick label style={font=\small},
        y tick label style={font=\footnotesize}
	]
        \addplot+[boxplot, fill, draw=black] table[y index=1] {\datatransposedA};
		\addplot+[boxplot, fill, draw=black] table[y index=1] {\datatransposedB};
	\end{axis}
\end{tikzpicture}
\caption{Final integrality gaps produced by $\mathcal{F}$ and $\mathcal{P}$.}\label{fig:boxplot}
\end{subfigure}
\begin{subfigure}[t]{0.45\textwidth}
\centering
\begin{tikzpicture}
    \pgfplotstableread[col sep=comma]{gap-heur-tg.csv}\csvdataC
    \pgfplotstabletranspose\datatransposedC{\csvdataC}

    \pgfplotstableread[col sep=comma]{gap-heur-r.csv}\csvdataD
	\pgfplotstabletranspose\datatransposedD{\csvdataD}

    \pgfplotstableread[col sep=comma]{gap-approx.csv}\csvdataE
	\pgfplotstabletranspose\datatransposedE{\csvdataE}
    
	\begin{axis}[
        ymode=log,
        log basis y={2},
        x=1cm,
        xminorgrids=true,
		boxplot/draw direction = y,
		x axis line style = {opacity=0},
		axis x line* = bottom,
		axis y line* = left,
        enlarge y limits={abs=0.4cm},
		ymajorgrids,
		xtick = {1, 2, 3, 4},
		xticklabels = {$\mathcal{H}$, $\mathcal{H}_R$, $\mathcal{A}$},
		xtick style = {draw=none}, 
		ylabel = {Gap \%},
        boxplot/whisker range={100},
        label style={font=\small},
        tick label style={font=\small},    
        title style={font=\small},
        legend style={font=\footnotesize},
        x tick label style={font=\small},
        y tick label style={font=\footnotesize},
        cycle list shift={2}
	]
        \addplot+[boxplot, fill, draw=black] table[y index=1] {\datatransposedC};
        \addplot+[boxplot, fill, draw=black] table[y index=1] {\datatransposedD};
        \addplot+[boxplot, fill, draw=black] table[y index=1] {\datatransposedE};
	\end{axis}
\end{tikzpicture}
\caption{Primal heuristic positive gaps (logarithmic scale).}\label{fig:boxplot:heur}
\end{subfigure}
\caption{Gaps produced by the proposed solving methods and heuristics.}\label{fig:boxplots}
\end{figure*}

Another distinguishing factor is the \R2{node-exploration} behavior in the proposed branch-and-price. 
The number of branch-and-bound nodes remains relatively low across the board, even for large \( n \) and high \( k \). This suggests that the strength of the \R2{root-node} relaxation allows $\mathcal{P}$ to make significant progress early in the search, avoiding large enumeration trees.


In summary, $\mathcal{F}$ is preferable for small-scale instances or when \( k \) is low, offering faster solutions and simpler implementation. 
However, as the problem {size} grows, $\mathcal{P}$ clearly demonstrates its superiority in both scalability and solution quality. Its consistent ability to solve more instances to optimality, maintain low gaps, and deliver stable computational performance makes it the more robust choice for larger and more difficult scenarios.

Finally, we observe that the primal heuristics based on the 
$k$-approximation algorithm presented in Section~\ref{sec:approximation} 
typically yield solutions with significantly better approximation ratios 
than those established in theory. 
Figure~\ref{fig:boxplot:heur} illustrates the quality of the solutions 
obtained by Algorithm~\ref{alg:heur}, denoted by~$\mathcal{H}$, by 
its variant that explores only the root node in 
Algorithm~\ref{alg:heur}, denoted by~$\mathcal{H}_R$, and by the $k$-approximation algorithm, denoted by~$\mathcal{A}$. 
We emphasize that these heuristics are independent of any specific 
formulation. 
The reported \emph{primal heuristic gaps} are computed with respect to 
the best primal solutions found by methods~$\mathcal{F}$ and~$\mathcal{P}$.

\section{Conclusion}
\label{sec:conclusion}
{
This work presents both theoretical and practical contributions to the balanced spanning forest problem.
On the theoretical side, we investigated its computational complexity by proving that both the min-max and max-min variants are $\NP$-hard even on highly restricted graph classes, including complete graphs with~$k=2$,  \R2{unweighted split graphs, and unweighted} bipartite graphs for any fixed $k\geq 2$.
Moreover, we showed that neither problem admits a subexponential-time algorithm unless the Exponential-Time Hypothesis fails.
We also designed a tight $k$-approximation algorithm for the min-max variant, which was subsequently used as the basis for the primal heuristics in our computational study.}

{
On the practical side, we proposed three mixed-integer linear programming formulations for the \textsc{min-max-bsf} problem and evaluated them through extensive computational experiments on synthetically generated instances.
The results highlight the strengths and limitations of each approach.
The compact flow formulation performs well on small instances and instances with $k=2$, but its performance deteriorates considerably as the problem size or the number of partitions increases.
In contrast, the branch-and-price algorithm consistently produced stronger lower bounds and solved substantially more challenging instances to optimality, particularly when $k>2$.}

{
The main limitation of the branch-and-price approach is the tailing-off effect during column generation, especially at the root node.
For some instances, improvements in the dual bound become very slow while the time required to solve the pricing problem increases significantly.
Consequently, improving the efficiency of the pricing algorithm appears to be 
a promising direction for future research.
Additional avenues include the development of stronger stabilization techniques, 
 more effective branching rules, and valid inequalities to further accelerate the solution process.}
\section*{Acknowledgments}
	Research of the second and  third authors was supported by Internal Funds KU Leuven (Grant C14/22/026).
    The second author was partially supported by CNPq-Brazil (Grant 404779/2025-5).

\bibliography{biblio}
\bibliographystyle{abbrvnat}
\end{document}